\documentclass[sigconf, nonacm]{acmart}

\usepackage{pvldb}

\renewcommand\vldbdoi{XX.XX/XXX.XX}
\renewcommand\vldbpages{XXX-XXX}
\renewcommand\vldbavailabilityurl{URL_TO_YOUR_ARTIFACTS}

\usepackage{amsmath}
\usepackage{amsthm}
\usepackage{algorithmic}
\usepackage{textcomp}
\usepackage[linesnumbered,ruled,vlined]{algorithm2e}
\usepackage{setspace}
\usepackage{graphicx}
\usepackage{multirow}
\usepackage{multicol}
\usepackage{makecell}
\usepackage{enumitem}
\usepackage{array}
\usepackage{pifont}
\usepackage{diagbox}
\usepackage{caption}
\usepackage{xcolor}
\usepackage{colortbl}
\usepackage{natbib}
\usepackage{footmisc}
\usepackage{subcaption}
\usepackage{siunitx}
\usepackage{url}
\usepackage[table]{xcolor}
\definecolor{best}{HTML}{CCE5F2}  
\definecolor{worst}{HTML}{F2E0CC}  
\definecolor{rankone}{HTML}{f29495}    
\definecolor{ranktwo}{HTML}{f5b5b5}    
\definecolor{rankthree}{HTML}{f8e5e5}  

\SetKwRepeat{Do}{do}{while}
\newcommand{\continue}{\textbf{continue}}

\newtheorem{myExample}{Example}
\newtheorem{myDef}{Definition}
\newtheorem{myTheorem}{Theorem}
\newtheorem{myLemma}{Lemma}
\newtheorem{myProblem}{Problem}

\begin{document}
\title{Common-Neighbor-Count-Based Representative Possible World Finding on Uncertain Graphs}

\author[Gu et al.]{%
  Chengjie Gu$^{1}$, Xiaoliang Xu$^{1}$, Yuxiang Wang$^{1}$, Kai Yao$^{2}$\\
  Mengzhao Wang$^{1}$, Tianxing Wu$^{3}$, Yingjie Xia$^{1}$, Xiangyu Ke$^{4}$%
}

\affiliation[obeypunctuation=true]{%
  \institution{$^1$Hangzhou Dianzi University, China; $^2$ Sichuan normal university, China;\\
  $^3$ Southeast University, China; $^4$ Zhejiang University, }%
  \country{China}%
}

\email{{chengjiegu,xxl,lsswyx,wmzssy}@hdu.edu.cn, kaiyao@sicnu.edu.cn, tianxingwu@seu.edu.cn,{xiayingjie,xiangyu.ke}@zju.edu.cn}

\begin{abstract}
Given an uncertain graph $\mathcal{G}$ with probabilistic edges, a representative possible world (RPW) is a deterministic graph that closely preserves the expected value of a designated basic structural feature in $\mathcal{G}$. As a proxy for $\mathcal{G}$, the RPW allows us to directly obtain approximate expectations of the higher-order features related to this preserved basic feature, thereby avoiding expensive enumeration or sampling on $\mathcal{G}$. Existing RPWs preserve per-node features such as degree and triangle degree but overlook common-neighbor counts between node pairs. As a result, they may poorly approximate the expected values of higher-order features that depend on these counts, e.g., Jaccard coefficient. Therefore, we study the \underline{C}ommon-neighbor-count-based \underline{R}epresentative \underline{P}ossible \underline{W}orld (CRPW) problem, extending RPWs from node features to node-pair features. Our problem seeks the possible world that best preserves the expected numbers of common neighbors between node pairs, and we prove that it is NP-hard. To address it, we develop a two-stage basic algorithm that quickly initializes a possible world and then refines it iteratively. We next accelerate the refinement with a counting-based estimation, as the refinement only requires determining whether an edge state change to the possible world is beneficial, rather than computing its exact magnitude. Furthermore, we design a Beta distribution-based adaptive termination method to automatically stop the refinement once the desired quality of the possible world is reached, preventing over- or under-execution. Extensive experiments on five real-world uncertain graphs demonstrate the effectiveness of our algorithms in preserving structural features. 
\end{abstract}

\maketitle
\vldbtopmatter

\section{Introduction}
\label{sec:introduction}

\begin{figure}[t]
  \setlength{\abovecaptionskip}{0.1cm}
    \centering
    \includegraphics[width=0.95\linewidth]{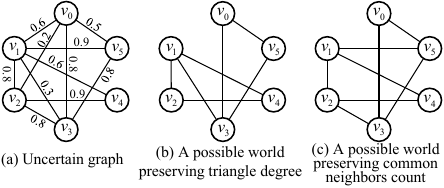}
    \caption{(a) An uncertain graph, (b-c) two of its possible worlds preserving the triangle degree and the number of common neighbors respectively.}
    \label{fig:intro_example}
    \vspace{-0.5cm}
\end{figure}

Graphs are widely used across scientific and industrial domains, where nodes represent objects and edges reflect the connections among them. In many real-world scenarios, such connections cannot be observed with certainty, but are instead estimated or predicted, and are thus inherently associated with an existence probability. In protein-protein interaction networks, an interaction between two proteins is established through noisy biological experiments and therefore holds only with a confidence derived from experimental evidence \cite{eronen2012biomine, asthana2004predicting}. In social networks, the influence or trust between two users is estimated from their interaction history, yielding a probabilistic connection \cite{klimt2004enron, KDD:Bonchi2014}. In knowledge graphs, relations are extracted by automated tools that often with errors, so each extracted edge carries a confidence score \cite{pujara2013knowledge, chen2019embedding}. To faithfully capture such uncertainty, \textit{uncertain graphs} \cite{parchas2014pursuit} are employed, in which each edge is annotated with an existence probability varies in $[0,1]$. Figure~\ref{fig:intro_example}(a) illustrates an uncertain graph with six nodes $\{v_0, \ldots, v_5\}$, where the number attached to each edge is its existence probability, e.g., the edge $e_{v_0 v_1}$ exists with probability of $0.6$. For clarity, edges with zero existence probability are often omitted.

To interpret the semantics of uncertain graphs, the \textit{possible world model} is widely adopted \cite{dai2021core, yang2019index, bonchi2014core, qiao2021maximal, zou2017truss}. Under this model, each edge with existence probability $p$ is independently observed to be either present with probability $p$ or absent with probability $1-p$. After observing all edges, we then yields a \textit{possible world}, i.e., a deterministic graph that retains exactly the edges observed to be present. Owing to the independence of the observations, such a possible world occurs with a probability equal to the product of existence probabilities of its present edges and non-existence probabilities of its absent ones. For example, Figure~\ref{fig:intro_example}(b) shows a possible world of the uncertain graph in Figure~\ref{fig:intro_example}(a), consisting of seven edges observed to be present while the other four edges are absent. The occurrence probability of such a possible world multiplies the existence probabilities $\{0.8,0.5,0.8,0.3,0.6,0.9,0.8\}$ of all present edges $\{e_{v_0v_3},e_{v_0v_5},e_{v_1v_2},e_{v_1v_3},e_{v_1v_4},e_{v_2v_4},e_{v_3v_5}\}$ and the non-existence probabilities $\{1-0.6,1-0.2,1-0.9,1-0.8\}$ of all absent edges $\{e_{v_0v_1},e_{v_0v_2},e_{v_1v_5},e_{v_2v_3}\}$, yielding the result of $2.65 \times 10^{-4}$.

\vspace{0.1cm}
Under the possible world model, an uncertain graph defines a probability distribution over its possible worlds. These possible worlds differ in topology, so the same structural feature may take different values when computed on different possible worlds~\cite{kaveh2021defining, saha2021shortest, peng2018efficient}. Consequently, existing studies characterize these structural features using their expected values over the distribution~\cite{pfeiffer2011methods, parchas2014pursuit}. Expectations of basic structural features are often directly computable from edge probabilities, e.g., a node's expected degree equals the sum of its incident-edge probabilities. However, expectations of higher-order features, such as shortest-path distance between two nodes, do not admit such straightforward computation. To compute them, the simplest method evaluates feature in every possible world and probabilistically aggregates their values. But this is prohibitive as the number of possible worlds grows exponentially with the number of edges \cite{jin2011discovering, ball1986computational}. To this end, sampling-based methods are proposed to estimate feature expectations \cite{jin2011distance,potamias2010knn}. However, achieving sufficient estimation accuracy often requires a large sample size, incurring substantial overhead \cite{caflisch1998monte}. Since higher-order features often depend on basic structural features, a more promising approach is to construct a \textit{representative possible world} (RPW) that preserves a designated basic structural feature \cite{parchas2014pursuit, song2016triangle}, allowing dependent higher-order features to be computed directly on the RPW.

\vspace{0.1cm}
\noindent{\textbf{Prior studies and motivation}.} Parchas et al.~\cite{parchas2014pursuit} first study the degree-based RPW (DRPW) that preserves the expected degree of each node. Their experiments show that DRPW approximates degree-related higher-order features, including shortest-path distance and betweenness centrality. However, Song et al. \cite{song2016triangle} note that node degree captures only the number of neighbors but overlooks triangle degree, another basic structural feature. They therefore construct the triangle-based RPW (TRPW) to additionally preserve each node's expected triangle count. Figure~\ref{fig:intro_example}(b) is such a TRPW obtained from Figure~\ref{fig:intro_example}(a). A triangle's occurrence probability is the product of its three edge probabilities. A node's expected triangle degree is the sum of the occurrence probabilities of all triangles containing that node. For node $v_5$ in Figure~\ref{fig:intro_example}(a), its expected triangle degree is $p(e_{v_5v_0})p(e_{v_5v_3})p(e_{v_0v_3})+p(e_{v_5v_0})p(e_{v_5v_1})p(e_{v_0v_1})+\dots=0.5\times0.8\times0.8+0.5\times0.6\times0.9+\dots=0.81$, while its triangle degree in Figure~\ref{fig:intro_example}(b) is $1$, yielding a triangle degree discrepancy of only $0.19$. Summing the absolute values of these discrepancies over all nodes yields a total triangle degree discrepancy of $1.31$ for this TRPW (Table~\ref{tab:intro_discrepancies}). Intuitively, triangle structure preserves common-neighbor relationships to some extent. Preserving triangle degree therefore shows potential for many common neighbor-related higher-order features, e.g., the Jaccard coefficient, which relies on the common-neighbor count of node pairs \cite{kogge2016jaccard}. However, the TRPW is prone to inaccuracy in some cases. A broken triangle no longer contributes to triangle degree, but its remaining edges may still form a common-neighbor relationship. Thus, preserving triangle degree does not ensure that TRPW faithfully reflects expected common-neighbor counts. We illustrate this limitation below.

\begin{table}[t]
  \centering
  \small
  \caption{Total discrepancies for three features}
    \label{tab:intro_discrepancies}

    \scalebox{0.89}{%
      \begin{tabular}{c||c|c}
        \textbf{Feature} & \textbf{TRPW (Fig.~\ref{fig:intro_example}(b))} & \textbf{CRPW (Fig.~\ref{fig:intro_example}(c))} \\ \hline\hline
        Triangle degree
          & 1.31 & 1.31 \\ \hline
        \# Common neighbors
          & 8.48 & 4.68\,($\downarrow$\,\textbf{44.8\%}) \\ \hline
        Jaccard
          & 2.53 & 1.12\,($\downarrow$\,\textbf{55.7\%})
      \end{tabular}%
    }

\end{table}

\begin{myExample}
We use node $v_2$ in Figure~\ref{fig:intro_example}(b) to illustrate TRPW's limitation. In Figure~\ref{fig:intro_example}(a), the expected triangle degree of $v_2$ is $0.848$. The TRPW retains triangle $(v_1,v_2,v_4)$, giving $v_2$ a triangle degree of $1$, close to its expectation. However, triangle $(v_0,v_1,v_2)$ is broken, so $v_2$'s triangle degree no longer captures the common-neighbor relationships within this triangle. For $v_0$ and $v_2$, the expected common-neighbor count in Figure~\ref{fig:intro_example}(a) is $p(e_{v_0v_1})p(e_{v_1v_2})+p(e_{v_0v_3})p(e_{v_2v_3})=0.6\times0.8+0.8\times0.8=1.12$, whereas their actual count in Figure~\ref{fig:intro_example}(b) is $0$, yielding a discrepancy of $1.12-0=1.12$. This shows that, despite preserving triangles, TRPW fails to faithfully reflect the expected number of common neighbors.
\end{myExample}

\vspace{-0.2cm}
Motivated by this limitation, we study \textit{\underline{C}ommon-neighbor-count-based \underline{R}epresentative \underline{P}ossible \underline{W}orld} (CRPW), which aims to directly preserve expected common-neighbor counts between node pairs in an uncertain graph $\mathcal{G}$. Figure~\ref{fig:intro_example}(c) shows the CRPW for Figure~\ref{fig:intro_example}(a). Given $v_0$ and $v_2$ with an expected common-neighbor count of $1.12$ in Figure~\ref{fig:intro_example}(a), the CRPW retains one common neighbor, yielding a discrepancy of $1.12-1=0.12$, whereas the TRPW in Figure~\ref{fig:intro_example}(b) retains none, causing a discrepancy of 
$1.12$. Aggregating such individual discrepancies across all node pairs yields the total common-neighbor discrepancy. In Table~\ref{tab:intro_discrepancies}, both CRPW and TRPW achieve the same total triangle degree discrepancy, but CRPW reduces the total common-neighbor discrepancy by $44.8\%$ compared to TRPW. This advantage also extends to the Jaccard coefficient that depends on common neighbors, enabling a discrepancy $55.7\%$ lower than TRPW.

\vspace{0.1cm}
\noindent\textbf{Challenges and our solutions.} we face two non-trivial challenges for solving the CRPW problem.

\noindent\underline{Challenge I}: \textit{How to efficiently construct a CRPW?}
A naive CRPW construction enumerates all possible worlds of an uncertain graph to find the one that best preserves expected common-neighbor counts. Since each edge exists independently, an uncertain graph with edge set $\mathcal{E}$ has $2^{|\mathcal{E}|}$ possible worlds, making exhaustive enumeration infeasible for large graphs. For instance, such brute-force enumeration on an uncertain graph with only $43$ edges took over $20$ days, confirming its impracticality. We further prove that the CRPW problem is NP-hard (\S \ref{sec:problem_complexity_analysis}), precluding an exact polynomial-time algorithm. Thus, efficiently constructing a high-quality CRPW poses a significant algorithmic challenge.

\noindent\underline{Solution I.}
To address this challenge, we develop a two-stage basic algorithm that quickly initializes a possible world and iteratively refines it (\S\ref{sec:basic_algorithm}). This paradigm's effectiveness has been validated in prior RPW studies \cite{parchas2014pursuit, song2016triangle}. 
During refinement, the core operation involves selecting an edge to precisely compute the variation in total discrepancy resulting from changing its presence state in the possible world, and using it to determine whether to apply such a state change. However, this exact computation is cumbersome, as it requires evaluating and aggregating the discrepancy variations for all node pairs affected by the edge state change. Worse still, iteratively repeating this process in refinement incurs expensive overhead. Instead of computing the exact variation of total discrepancy, determining its trend, i.e., whether it decreases or increases, suffices to evaluate the benefit of the edge state change, offering an efficient alternative. Hence, in \S \ref{sec:integer_based_refinement}, we replace exact variation computation with counting-based trend estimation. Specifically, we round discrepancies to integers and theoretically derive a counting criterion that assesses an edge state change by checking whether affected node pairs contributing positively outnumber those contributing negatively.

\vspace{0.1cm}
\noindent\underline{Challenge II}: \textit{How to terminate refinement once the desired quality is reached?}
Iterative refinement requires a termination condition. Existing methods \cite{parchas2014pursuit,song2016triangle} typically require users to preset the number of iterations $k$. However, choosing an appropriate $k$ is inherently challenging. The user sets $k$ with the aim of obtaining a possible world with desired quality. An overly large $k$ incurs unnecessary computational overhead after the target quality is reached, whereas an overly small $k$ terminates refinement before that quality is attained. Moreover, even for the same desired quality, the required value of $k$ varies across datasets. This motivates a critical need for a dataset-agnostic automatic termination mechanism.

\noindent\underline{Solution II.}
Although the appropriate iteration count $k$ is difficult to set and varies across datasets, the objective of refinement is always to obtain a possible world that meets the user's desired quality. Therefore, refinement quality provides an intuitive way for determining termination. 
Since our refinement improves quality by reducing total discrepancy through edge state changes, the state-changeable edge proportion (SEP), defined as the fraction of edges whose state changes further reduce the total discrepancy, serves as an effective proxy for the quality of the current possible world (\S\ref{sec:beta_based_termination}). Intuitively, refinement should terminate when SEP drops below a user-specified quality threshold, indicating that few edges can still contribute to quality improvement.
However, computing SEP requires scanning all edges, which is computationally expensive. We therefore estimate SEP via uniform edge sampling, modeling each sampled edge as a Bernoulli trial. We then leverage the Beta distribution to compute the probability that the true SEP falls below the quality threshold, terminating refinement when this probability exceeds a confidence level. This ensures efficient, adaptive termination across datasets.

\begin{table}[t]
  \small
  \caption{Comparison of representative methods and ours}
    \label{tab:method_comparison}
    \setlength{\tabcolsep}{3pt}
    \resizebox{\columnwidth}{!}{

    \begin{tabular}{c||c|c|c}
      \textbf{RPW}  & \textbf{Feature}         & \textbf{Discrepancy evaluation} & \textbf{Termination criterion}   \\ \hline\hline
      DRPW & Degree          & Tardy exact comp.    & Fixed preset iteration   \\ \hline
      TRPW & Triangle Degree & Tardy exact comp.    & Fixed preset iteration \\ \hline
      \rowcolor{rankthree!40} CRPW & \#Common neighbor & Fast trend est.       & Quality-driven adaptation
    \end{tabular}
    }
\end{table}

\vspace{0.1cm}
\noindent\textbf{Contribution.} 
Table~\ref{tab:method_comparison} summarizes the key differences between our CRPW and existing RPWs. Our contributions are:

\begin{itemize}
\item We formally define the CRPW problem that preserves expected common-neighbor counts of uncertain graphs (\S\ref{sec:problem_definition}) and prove its NP-hardness (\S\ref{sec:problem_complexity_analysis}).
\item We present a two-stage basic algorithm that initializes a possible world and iteratively refines it (\S\ref{sec:basic_algorithm}), and propose a counting-based trend estimation method to avoid tardy exact computation of total discrepancy variations (\S\ref{sec:integer_based_refinement}).
\item We introduce SEP as a measure of refinement quality and estimate it from edge samples using a Beta distribution, enabling adaptive termination once the user-desired quality is attained with sufficient confidence (\S\ref{sec:beta_based_termination}).
\item We conduct extensive experiments to evaluate the effectiveness (\S \ref{sec:Comparison with Existing RPW Methods}, \S\ref{sec:Comparison with the Exact CRPW}), efficiency (\S \ref{sec:exp_efficiency}), parameter sensitivity (\S \ref{sec:exp_sensitivity}) of algorithms, as well as the effect of Beta-based adaptive termination (\S \ref{sec:exp_beta}, \S\ref{sec:Practical vs. Theory-Aligned Sampling}).
\end{itemize}

\section{Preliminaries}
\label{sec:pre}

\subsection{Problem Definition}
\label{sec:problem_definition}

\begin{myDef}[\textbf{Uncertain Graph} \cite{parchas2014pursuit}]
	\label{def:uncertain_graph}
	An uncertain graph $ \mathcal{G} = \left(\mathcal{V}, \mathcal{E}, \mathcal{P}\right) $ involves a node set $\mathcal{V}$, an edge set $\mathcal{E}$, and a function $\mathcal{P}:\mathcal{E}\to(0,1]$ that assigns an existence probability $p(e_{uv})\in (0,1]$ to each undirected edge $e_{uv} \in \mathcal{E}$ between two nodes $u$ and $v$ in $\mathcal{V}$. For simplicity, we omit the subscript of an edge when necessary.
\end{myDef}

Under the context of uncertain graph, observing any probabilistic edge $e\in\mathcal{E}$ yields two outcomes: the edge is present with probability $p(e)$ or absent with probability $1-p(e)$. As a result, the topology of uncertain graph is nondeterministic and can change across different observations. To model this uncertainty and faithfully express the semantics of uncertain graphs, researchers commonly adopt the \textit{possible world model} \cite{dai2021core,yang2019index,bonchi2014core,peng2018efficient,qiao2021maximal,zou2017truss}, with the definition below.

\begin{myDef}[\textbf{Possible World} \cite{parchas2014pursuit}]
  \label{def:possible_world}
	Given an uncertain graph $\mathcal{G} = (\mathcal{V}, \mathcal{E}, \mathcal{P})$, a possible world $G = (V, E)$ of $\mathcal{G}$ is a deterministic graph derived from $\mathcal{G}$, where $V = \mathcal{V}$ and $E \subseteq \mathcal{E}$. 
\end{myDef}

With the possible world model, the independence of edge existence is a widely adopted assumption \cite{jin2011discovering,jin2011distance,potamias2010k}, facilitating the calculation of the existing probability of a specific possible world $G$. We denote the event that the uncertain graph $\mathcal{G}$ is observed as one possible world $G$ by $G \sqsubseteq \mathcal{G}$. The probability of the event $G \sqsubseteq \mathcal{G}$ is then calculated by Eq. \ref{eq:possible_world}, denoted by $P\{G\sqsubseteq \mathcal{G}\}$ (shorten as $P\{G\}$ when no ambiguity arises).

\begin{equation}
		\label{eq:possible_world}
		P\{G \sqsubseteq \mathcal{G}\} = \prod_{e \in E} p(e) \prod_{e\in \mathcal{E} \setminus E} (1 - p(e))
\end{equation}

We next define common-neighbor counts and the associated discrepancy measures to formulate the \textit{\underline{C}ommon-neighbor-count-based \underline{R}epresentative \underline{P}ossible \underline{W}orld} (CRPW) problem.

\begin{myDef}[\textbf{Common Neighbors}]
	\label{def:cn}
	Given a possible world $G = (V, E)$ and two nodes $u, v \in V$, the common neighbors between $u$ and $v$ in $G$, denoted as $C_G(u, v)$, is defined as Eq. \ref{eq:cn}, where $N_G(u)$ ($N_G(v)$) is the set of neighbors of node $u$ ($v$) in $G$.
	\begin{equation}
		\label{eq:cn}
		C_G(u, v) = N_G(u) \cap N_G(v)
	\end{equation}
\end{myDef}

For nodes $u$ and $v$ in a possible world $G$, we denote the number of their common neighbors by $|C_G(u, v)|$. In the uncertain graph, this quantity becomes a random variable over possible worlds. We therefore consider its expectation, defined as follows.

\begin{myDef}[\textbf{Expected Number of Common Neighbors}]
	\label{def:exp_cn}
  Given an uncertain graph $\mathcal{G} = (\mathcal{V}, \mathcal{E}, \mathcal{P})$ and two nodes $u,v \in \mathcal{V}$, their expected number of common neighbors $\overline{C}_{\mathcal{G}}(u, v)$ is shown in Eq. \ref{eq:exp_cn_computation}, where $C_{\mathcal{G}}(u,v)$ denotes the set of common neighbors of $u$ and $v$ within the topological structure of $\mathcal{G}$, ignoring edge probabilities. 
  \begin{equation}
    \label{eq:exp_cn_computation}
		\overline{C}_{\mathcal{G}}(u, v) = \sum_{w \in C_{\mathcal{G}}(u, v)}p(e_{uw})\cdot p(e_{vw})
	\end{equation}
\end{myDef}

Intuitively, a possible world $G$ is a representative instance of an uncertain graph $\mathcal{G}$ regarding common neighbors if the count for any two nodes $u$ and $v$ in $G$ closely approximates its expected value in $\mathcal{G}$, that is, the discrepancy between them is small. Since nodes $u$ and $v$ uniquely determine the edge $e_{uv}$, we term this discrepancy in their common neighbors as the \textit{edge discrepancy} for convenience.

\begin{myDef}[\textbf{Edge Discrepancy}]
	\label{def:edge_dis}
	Given an uncertain graph $\mathcal{G} = (\mathcal{V}, \mathcal{E}, \mathcal{P})$, a possible world $G \sqsubseteq \mathcal{G}$, and an edge $e_{uv} \in \mathcal{E}$, the edge discrepancy of $e_{uv}$ is defined as the difference between the number of common neighbors of $u$ and $v$ in $G$ and their expected number of common neighbors in $\mathcal{G}$, formally expressed as below.
	\begin{equation}
		\label{eq:edge_dis}
		dis_{G \sqsubseteq \mathcal{G}} (e_{uv}) = |C_{G}(u,v)| - \overline{C}_{\mathcal{G}}(u, v)
	\end{equation}
\end{myDef}

For brevity, we denote $dis_{G \sqsubseteq \mathcal{G}} (e_{uv})$ by $dis_{G} (e_{uv})$ when the context is clear. Based on Definition \ref{def:edge_dis}, we further define the total discrepancy of a possible world $G$ w.r.t. an uncertain graph $\mathcal{G}$.

\begin{myDef}[\textbf{Total Discrepancy}]
	\label{def:tcn_dis}
	Given an uncertain graph $\mathcal{G} = (\mathcal{V}, \mathcal{E}, \mathcal{P})$, a possible world $G \sqsubseteq \mathcal{G}$, the total discrepancy of $G$, denoted as $dis_{\mathcal{G}}(G)$, is defined as the sum of the absolute values of the edge discrepancy over all edges in $\mathcal{G}$, given by Eq. \ref{eq:tcn_dis}.
  \begin{equation}
		\label{eq:tcn_dis}
		dis_{\mathcal{G}}(G) = \sum_{e \in \mathcal{E}} |dis_{G \sqsubseteq \mathcal{G}}(e)|
	\end{equation}
\end{myDef}

Total discrepancy provides a quantitative measure of how well a possible world $G$ preserves the expected number of common neighbors in the uncertain graph $\mathcal{G}$. Accordingly, we seek a \textit{\underline{C}ommon-neighbor-count-based \underline{R}epresentative \underline{P}ossible \underline{W}orld} (CRPW) that minimizes this total discrepancy. A smaller discrepancy indicates a higher-quality representative possible world. Formally, the CRPW problem is defined as follows.

\begin{myProblem}[\textbf{CRPW Problem}]
	\label{prob:CNRPW}
	Given an uncertain graph $\mathcal{G}=(\mathcal{V}, \mathcal{E}, \mathcal{P})$, the CRPW Problem aims to find a possible world $G\sqsubseteq  \mathcal{G}$ such that it has the smallest total discrepancy w.r.t. $\mathcal{G}$, given by Eq. \ref{eq:CNRPW}.
    \begin{equation}
		\label{eq:CNRPW}
	G = \mathop{\arg\min}\limits_{G \sqsubseteq  \mathcal{G}}dis_{\mathcal{G}}(G)
  \end{equation}
\end{myProblem}

\subsection{Analysis of Problem Complexity}
\label{sec:problem_complexity_analysis}

We establish that the CRPW Problem is NP-hard. Since NP-hardness is defined over decision problems, we investigate the decision version of the CRPW Problem, which is formulated below.

\begin{myProblem}[\textbf{CRPW Decision Problem}]
  \label{prob:CRPW_DEC}
  Given an uncertain graph $\mathcal{G^*}=(\mathcal{V^*}, \mathcal{E^*}, \mathcal{P^*})$ and a threshold $T$, the CRPW Decision Problem determines whether there exists a possible world $G^* \sqsubseteq \mathcal{G^*}$ such that the total discrepancy $dis_{\mathcal{G^*}}(G^*) \le T$.
\end{myProblem}

Generally, the decision version of a problem is no harder than its optimization version \cite{kleinberg2006algorithm}: any polynomial-time algorithm for the CRPW Problem also solves the CRPW Decision Problem, by computing the total discrepancy of the returned possible world and comparing it against $T$. Thus, if the CRPW Decision Problem is NP-hard, then the CRPW Problem is NP-hard as well. Therefore, we focus on proving the NP-hardness of the CRPW Decision Problem.

\begin{myTheorem}
  \label{th:np_hard}
  The CRPW Decision Problem is NP-hard.
\end{myTheorem}

\begin{proof}
We prove the NP-hardness of the CRPW Decision Problem by a reduction from the Monotone 1-in-3 3SAT problem, a well-known NP-complete problem~\cite{schaefer1978complexity}. An instance of Monotone 1-in-3 3SAT consists of $y$ Boolean variables $B_1, \dots, B_y$ and $z$ clauses $A_1, \dots, A_z$, where each clause contains exactly three positive variables (no negations). The Monotone 1-in-3 3SAT instance is satisfiable if there exists a truth assignment such that every clause has exactly one true variable. Our reduction maps a Monotone 1-in-3 3SAT instance to an uncertain graph $\mathcal{G}^*$ with a threshold $T$, such that the Monotone 1-in-3 3SAT instance is satisfiable if and only if $(\mathcal{G}^*, T)$ is a Yes-instance of the CRPW Decision Problem. The two sides are stated in different languages: one is Boolean satisfiability, the other is the minimum total discrepancy over possible worlds. To establish this equivalence, we proceed in the following four steps.

\vspace{0.1cm}
\noindent\textbf{Step 1. Construction of the uncertain graph $\mathcal{G}^*$.} Given a Monotone 1-in-3 3SAT instance with boolean variables $B_1, \dots, B_y$ and clauses $A_1, \dots, A_z$, we construct the uncertain graph $\mathcal{G}^* = (\mathcal{V}^*, \mathcal{E}^*, \mathcal{P}^*)$ as follows. (1) For node set $\mathcal{V}^*$, it consists of: a center node $c$; a boolean variable node $b_i$ for each boolean variable $B_i$; and a clause node $a_j$ for each clause $A_j$. (2) For edge set $\mathcal{E}^*$, it consists of three types of edges. A variable edge $e_{cb_i}$ connects the center node $c$ and the boolean variable node $b_i$. A clause edge $e_{ca_j}$ connects the center node $c$ and the clause node $a_j$. A connection edge $e_{a_jb_i}$ connects the clause node $a_j$ and the boolean variable node $b_i$, and is present only when boolean variable $B_i$ appears in clause $A_j$. (3) For existence probabilities $\mathcal{P}^*$, we set $p(e_{cb_i}) = \frac{1}{3}$ for every variable edge and $p(e_{ca_j}) = p(e_{a_jb_i}) = 1$ for every clause edge and connection edge. The presence of a variable edge in a possible world encodes $B_i = \text{True}$, and its absence encodes $B_i = \text{False}$. We set its probability to $\frac{1}{3}$ so that the expected number of common neighbors between the endpoints of each clause edge $e_{ca_j}$ equals $3 \times \frac{1}{3} = 1$. This encodes the 1-in-3 condition of clause $A_j$ into the discrepancy of $e_{ca_j}$. The clause edges and connection edges are assigned probability $1$, so they are present in every possible world; therefore possible worlds differ only in which variable edges they retain. As a result, the constructed uncertain graph $\mathcal{G}^*$ has $1 + y + z$ nodes and $y + z + 3z = y + 4z$ edges, with all existence probabilities equal to $\frac{1}{3}$ or $1$. Hence, $\mathcal{G}^*$ is constructed in polynomial time w.r.t. the size of the Monotone 1-in-3 3SAT instance, and the reduction is valid.

\vspace{0.1cm}
\noindent\textbf{Step 2. Total discrepancy of $G^* \sqsubseteq \mathcal{G}^*$.} We next show the total discrepancy calculation of a possible world $G^* \sqsubseteq \mathcal{G}^*$ with an algebraic function. Since possible worlds differ only in which variable edges they retain, we encode this difference by an indicator $x_i \in \{0,1\}$ for each variable edge $e_{cb_i}$, where $x_i = 1$ if $e_{cb_i}$ is present in $G^*$ and $x_i = 0$ otherwise. Next, for each edge type, we compute the expected number of common neighbors of its endpoints in $\mathcal{G}^*$ by Definition~\ref{def:exp_cn}, then express the edge discrepancy by Definition~\ref{def:edge_dis}. (1) For clause edge $e_{ca_j}$. In the topology of $\mathcal{G}^*$, each boolean variable node $b_i$ with $B_i \in A_j$ is a common neighbor of $c$ and $a_j$. Thus $\overline{C}_{\mathcal{G}^*}(c, a_j) = \sum_{B_i \in A_j} p(e_{cb_i})\, p(e_{a_jb_i}) = 3 \times \frac{1}{3} \times 1 = 1$. In $G^*$, such a $b_i$ is a common neighbor only if its variable edge $e_{cb_i}$ is present, i.e., $x_i = 1$, so the number of common neighbor of $a_j$ and $b_i$ is $\sum_{i \in A_j} x_i$. Hence the edge discrepancy of $e_{ca_j}$ is $dis_{G^*}(e_{ca_j}) = \sum_{i \in A_j} x_i - 1 $. (2) For connection edge $e_{a_jb_i}$. In the topology of $\mathcal{G}^*$, the only common neighbor of $a_j$ and $b_i$ is the center node $c$. Thus $\overline{C}_{\mathcal{G}^*}(a_j, b_i) = p(e_{ca_j})\, p(e_{cb_i}) = \frac{1}{3}$. In $G^*$, center node $c$ is a common neighbor only if the variable edge $e_{cb_i}$ is present, i.e., $x_i = 1$, so the number of common neighbor of $a_j$ and $b_i$ is $x_i$. Hence the edge discrepancy of $e_{a_jb_i}$ is $dis_{G^*}(e_{a_jb_i}) =  x_i - \frac{1}{3} $. (3) For variable edge $e_{cb_i}$. In the topology of $\mathcal{G}^*$, each clause node $a_j$ whose clause $A_j$ contains $B_i$ is a common neighbor of $c$ and $b_i$. Thus $\overline{C}_{\mathcal{G}^*}(c, b_i) = \sum_{j: B_i \in A_j} p(e_{ca_j})\, p(e_{a_jb_i}) = \sum_{j: B_i \in A_j} 1$, which means the number of clauses containing $B_i$. In $G^*$, each such $a_j$ is connected to $c$ by $e_{ca_j}$ and to $b_i$ by $e_{a_jb_i}$, both of which are present, so every such $a_j$ remains a common neighbor. Therefore the number of common neighbor of $c$ and $b_i$ is the number of clauses containing $B_i$, i.e., $\sum_{j: B_i \in A_j} 1$. So the edge discrepancy of $e_{cb_i}$ is $dis_{G^*}(e_{cb_i}) = 0$. 

Summing the absolute edge discrepancies over all edges (Definition~\ref{def:tcn_dis}), we obtain the total discrepancy of $G^* \sqsubseteq \mathcal{G}^*$ as follows:
\begin{equation}
  \label{eq:np_total_two_terms}
  dis_{\mathcal{G}^*}(G^*) = \sum_{j=1}^z \Big| \sum_{i \in A_j} x_i - 1 \Big| + \sum_{j=1}^z \sum_{i \in A_j} \Big| x_i - \frac{1}{3} \Big| + \sum_{i=1}^y 0,
\end{equation}
where the three terms come from the clause edges, the connection edges, and the variable edges, respectively.

We next simplify the Eq.\ref{eq:np_total_two_terms}. Since $x_i \in \{0,1\}$, we have $|x_i - \frac{1}{3}| = \frac{1}{3} + \frac{1}{3} x_i$: it equals $\frac{1}{3}$ when $x_i = 0$ and $\frac{2}{3}$ when $x_i = 1$. As each clause $A_j$ contains exactly three boolean variables, summing this over the boolean variables in $A_j$ gives $\sum_{i \in A_j} (\frac{1}{3} + \frac{1}{3} x_i) = 1 + \frac{1}{3} \sum_{i \in A_j} x_i$.

Substituting this into Eq.~\ref{eq:np_total_two_terms} and grouping all terms by clause, the total discrepancy becomes a sum of per-clause costs,
\begin{equation}
  \label{eq:np_clause_cost}
  dis_{\mathcal{G}^*}(G^*) = \sum_{j=1}^z Cost(A_j), \quad
  Cost(A_j) = \Big| \sum_{i \in A_j} x_i - 1 \Big| + 1 + \frac{1}{3} \sum_{i \in A_j} x_i,
\end{equation}
where $Cost(A_j)$ is the cost contributed by clause $A_j$, and $\sum_{i \in A_j} x_i$ is the number of boolean variables in clause $A_j$ assigned true. Each $Cost(A_j)$ depends only on the boolean variables of clause $A_j$, so the clauses are decoupled and the total discrepancy can be minimized by minimizing each $Cost(A_j)$ independently.

\vspace{0.1cm}
\noindent\textbf{Step 3. Minimum total discrepancy.} We now determine the minimum total discrepancy over all possible worlds of $\mathcal{G}^*$. By Eq.~\ref{eq:np_clause_cost}, the total discrepancy is the sum of the per-clause costs $Cost(A_j)$, each minimized independently. We therefore first determine the minimum of a single $Cost(A_j)$. Since $Cost(A_j)$ depends only on the number of boolean variables in clause $A_j$ assigned true, i.e., $\sum_{i \in A_j} x_i \in \{0,1,2,3\}$, we enumerate the four cases: $Cost(A_j) = 2, \frac{4}{3}, \frac{8}{3}, 4$ when $\sum_{i \in A_j} x_i = 0, 1, 2, 3$, respectively. Hence $Cost(A_j)$ attains its minimum $\frac{4}{3}$ uniquely when $\sum_{i \in A_j} x_i = 1$, which obtained Eq. \ref{eq:np_clause_bound} with equality if and only if exactly one boolean variable in clause $A_j$ is assigned true.
\begin{equation}
  \label{eq:np_clause_bound}
  Cost(A_j) \ge \frac{4}{3}
\end{equation}

Summing Eq.~\ref{eq:np_clause_bound} over all $z$ clauses gives a lower bound on the total discrepancy of any possible world, shown in Eq. \ref{eq:np_global_bound}.
\begin{equation}
  \label{eq:np_global_bound}
  dis_{\mathcal{G}^*}(G^*) = \sum_{j=1}^z Cost(A_j) \ge \frac{4z}{3}
\end{equation}

\vspace{0.1cm}
\noindent\textbf{Step 4. Reduction correctness.} We set the threshold $T$ of the CRPW Decision Problem to $T = \frac{4z}{3}$, the minimum total discrepancy established in Eq.~\ref{eq:np_global_bound}. We show that $(\mathcal{G}^*, T)$ is a Yes-instance, i.e., there exists a possible world $G^*$ with $dis_{\mathcal{G}^*}(G^*) \le \frac{4z}{3}$, if and only if the Monotone 1-in-3 3SAT instance is satisfiable. Since $dis_{\mathcal{G}^*}(G^*) \ge \frac{4z}{3}$ holds for every possible world by Eq.~\ref{eq:np_global_bound}, the condition $dis_{\mathcal{G}^*}(G^*) \le \frac{4z}{3}$ is equivalent to $dis_{\mathcal{G}^*}(G^*) = \frac{4z}{3}$.

\vspace{0.05cm}
($\Rightarrow$) Suppose the Monotone 1-in-3 3SAT instance is satisfiable. Then there is a truth assignment under which every clause has exactly one true variable. We construct a possible world $G^*$ by setting $x_i = 1$ if $B_i$ is true and $x_i = 0$ otherwise. Under this assignment, each clause has exactly one true variable, so $Cost(A_j) = \frac{4}{3}$ for every clause by Eq.~\ref{eq:np_clause_bound}, and thus $dis_{\mathcal{G}^*}(G^*) = \frac{4z}{3}$. Hence $(\mathcal{G}^*, T)$ is a Yes-instance of the CRPW Decision Problem.

\vspace{0.05cm}
($\Leftarrow$) Suppose there is a possible world $G^*$ with $dis_{\mathcal{G}^*}(G^*) = \frac{4z}{3}$. Since $dis_{\mathcal{G}^*}(G^*) = \sum_{j=1}^z Cost(A_j)$ and each $Cost(A_j) \ge \frac{4}{3}$, the sum equals $\frac{4z}{3}$ only if $Cost(A_j) = \frac{4}{3}$ for every clause. This forces every clause to have exactly one true variable under the assignment defined by $x_i$, which thus satisfies the Monotone 1-in-3 3SAT instance.

\vspace{0.05cm}
Combining both directions, the constructed CRPW Decision Problem instance $(\mathcal{G}^*, T)$ is a Yes-instance if and only if the Monotone 1-in-3 3SAT instance is satisfiable. Since the construction takes polynomial time and Monotone 1-in-3 3SAT is NP-complete, the CRPW Decision Problem is NP-hard. As established earlier, the NP-hardness of the CRPW Decision Problem implies that of the CRPW Problem.
\end{proof}

The NP-hardness established above implies that finding an optimal possible world is intractable. We therefore resort to heuristics that efficiently produce a representative possible world in practice, as presented in the following sections.

\section{Basic Algorithm}
\label{sec:basic_algorithm}

To obtain a CRPW, a simple yet effective approach is to initialize a possible world from the uncertain graph  and iteratively refine its total discrepancy. Therefore, in this section, we propose a two-stage basic algorithm, consisting of (1) \textit{probability-based initialization} (in \ref{sec:probability_based_initial}) and (2) \textit{random selection-based refinement} (in \ref{sec:random_select_based_opt}).

\subsection{Probability-based Initialization}
\label{sec:probability_based_initial} 

A straightforward approach is to randomly select a subset of edges from the uncertain graph $\mathcal{G}$ to form a possible world $G$. However, it fails to capture the expected topological structure of $\mathcal{G}$ and typically results in a large total discrepancy $dis_{\mathcal{G}}(G)$, thereby requiring more iterations to reduce $dis_{\mathcal{G}}(G)$. To avoid this, we adopt a probability-based initialization strategy, in which each edge $e$ is independently sampled according to its existence probability $p(e)$, ensuring the sampling of $e$ follows a \textit{Bernoulli distribution} with parameter $p(e)$, where the edge exists with $p(e)$ and does not exist with $1-p(e)$.

\vspace{0.1cm}
\noindent\textbf{Overview}.
Algorithm \ref{alg:pin} outlines the process of probability-based initialization. For each edge $e \in \mathcal{E}$, we use the acceptance-rejection method \cite{flury1990acceptance} to simulate its independent sampling process (as justified in \textbf{Theorem \ref{th:bernoulli_distribution}}), where $e$ is accepted with $p(e)$ and rejected with $1-p(e)$ (lines 3-7). After all edges are processed, those accepted edges are collected to form the initial possible world.

\begin{algorithm}[t]
  \small
  \setstretch{0.85}
  \caption{\texttt{Probability-based Initialization}}
  \label{alg:pin}
  \LinesNumbered
  \KwIn{uncertain graph $\mathcal{G}=(\mathcal{V}, \mathcal{E}, \mathcal{P})$}
  \KwOut{an initial possible world $G_{0} \sqsubseteq \mathcal{G}$}

  $E \leftarrow \emptyset$, $V \leftarrow \mathcal{V}$\;
  \For{$e \in \mathcal{E}$}{
    generate a random number $r$ uniformly in [0, 1]\;
	\tcp{{\footnotesize the acceptance-rejection mechanism}}
    \If{$r \leq p(e)$}{
      add $e$ to $E$\tcc*{{\footnotesize accept $e$}}
    }
    \Else{
      \continue \tcc*{{\footnotesize reject $e$}}
    }

  }

  \Return $G_{0} \leftarrow (V, E)$\;
\end{algorithm}

In Algorithm~\ref{alg:pin}, we first initialize the possible world $G$ by setting its edge set $E$ to be empty and its node set $V$ to $\mathcal{V}$ (line~1). Then, for each edge $e \in \mathcal{E}$, we perform a sampling process as follows (lines 2-7): a random number $r$ is generated uniformly from the interval [0,1] (line 3). If $r \leq p(e)$, the edge is accepted and added to $E$ (lines 4-5); otherwise, the edge is rejected and excluded from $E$ (lines 6-7). After all edges have been processed, it returns the initial possible world $G_0 = (V, E)$ (line 8). Next, we explain the validity of using such a sampling process to simulate a Bernoulli distribution with parameter $p(e)$.

\begin{myTheorem}
  \label{th:bernoulli_distribution}
  Given an edge $e$ with existence probability $p(e)$, let $r \sim U(0,1)$. We define a discrete random variable $X_e$ and set $X_e = 1$ if $r \leqslant p(e)$, and $X_e = 0$ if $r>p(e)$. Then $X_e$ follows a Bernoulli distribution with parameter $p(e)$.
\end{myTheorem}

\begin{proof}
  \label{prf:bernoulli_distribution}
Since $r$ is drawn uniformly from $U[0,1]$, the probability that $X_e = 1$ and $X_e = 0$ can be calculated by Eq. \ref{eq:bernoulli_distribution=1} and Eq. \ref{eq:bernoulli_distribution=0}. Thus $X_e$ follows a Bernoulli distribution with parameter $p(e)$.
\vspace{-0.1cm}
\begin{equation}
		\label{eq:bernoulli_distribution=1}
		P(X_e=1)=P(r\le p(e))=\int_0^{p(e)} 1\,dx=p(e)
	\end{equation}
\vspace{-0.2cm}
\begin{equation}
		\label{eq:bernoulli_distribution=0}
		P(X_e=0)=P(r>p(e))=\int_{p(e)}^{1} 1\,dx=1-p(e)
	\end{equation}
\end{proof}

\subsection{Random Selection-based Refinement}
\label{sec:random_select_based_opt}
After obtaining the initial possible world $G_0$, we iteratively refine it by reducing the total discrepancy until a low-discrepancy CRPW $G$ is returned. To this end, we adopt a random selection-based refinement implemented in an iterative manner. Specifically, in the $i$-th iteration, we randomly select an edge $e \in \mathcal{E}$ and consider changing its existence state in the current possible world $G_i$, which yields a new possible world $G_{i+1}$. To determine whether such a change is beneficial, we evaluate the difference in total discrepancy between $G_{i+1}$ and $G_i$, referred to as the \textit{total discrepancy variation} w.r.t. an edge $e$, which is formally defined as follows. If this variation is less than $0$, the state of $e$ is changed, i.e., e is removed if $e \in G_i$, and added if $e \notin G_i$, resulting in $G_{i+1}$; otherwise, $e$ remains unchanged.

\begin{myDef}[\textbf{Total Discrepancy Variation}]
	\label{def:tcn_var}
	Given an uncertain graph $\mathcal{G} = (\mathcal{V}, \mathcal{E}, \mathcal{P})$, a possible world $G_i \sqsubseteq \mathcal{G}$ at the $i$-th iteration, and an edge $e \in \mathcal{E}$, let $G_{i+1}$ denote the possible world obtained by changing the existence state of $e$ in $G_i$. The total discrepancy variation of $e$ in $G_i$, denoted as $\delta_{G_i \sqsubseteq \mathcal{G}}(e)$, is defined as the difference in total discrepancy between $G_{i+1}$ and $G_i$, as shown in Eq. \ref{eq:tcn_dis_var}.
  \begin{equation}
		\label{eq:tcn_dis_var}
		\delta_{G_i \sqsubseteq \mathcal{G}}(e) = dis_{\mathcal{G}}(G_{i+1}) - dis_{\mathcal{G}}(G_i)
	\end{equation}
\end{myDef}

For brevity, we denote $\delta_{G_i \sqsubseteq \mathcal{G}}(e)$ by $\delta_{G_i}(e)$ when context is clear.

\noindent\textbf{Overview}. Given an uncertain graph $\mathcal{G}$ and an initial possible world $G_0$, the random selection-based refinement proceeds iteratively over $k$ iterations. In the $i$-th iteration, we randomly select an edge $e_{uv} \in \mathcal{E}$, evaluate whether changing its existence state in the current possible world $G_i$ reduces the total discrepancy (i.e., $\delta_{G_i}(e_{uv}) < 0$), and if so, change the state of $e_{uv}$ to obtain $G_{i+1}$. The key of random selection-based refinement is the efficient computation of $\delta_{G_i}(e_{uv})$, which faces a challenge: Directly applying Eq.~\ref{eq:tcn_dis_var} requires traversing all edges in $\mathcal{E}$ to compute their edge discrepancies, whereas changing the state of $e_{uv}$ only affects the common neighbors associated with a subset of edges, and hence only the discrepancies of those edges change; traversing the remaining edges introduces unnecessary computation. To address this challenge, we present an efficient computation of the total discrepancy variation $\delta_{G_i}(e_{uv})$, which restricts the computation to the subset of edges whose discrepancies are affected by the state change of $e_{uv}$. Building upon this optimization, we then present the algorithm of random selection-based refinement.

\vspace{0.1cm}
\noindent\textbf{Efficient Computation of Total Discrepancy Variation}. 
To efficiently compute the total discrepancy variation $\delta_{G_i}(e_{uv})$ for the edge $e_{uv}$, we introduce an \textit{edge influence set} that identifies edges whose discrepancies are affected by the state change of $e_{uv}$ (\textbf{Theorem \ref{th:eis}}). By aggregating the variation of absolute edge discrepancy only for edges from the edge influence set, we can efficiently obtain $\delta_{G_i}(e_{uv})$ (\textbf{Theorem \ref{th:tcn_dis_var_opt1}}).

Next, we begin by establishing how the state change of $e_{uv}$ affects the number of common neighbors between the endpoints of adjacent edges (Lemma~\ref{lem:cn_change}), which underlies the identification of edge influence set.

\begin{myLemma}
  \label{lem:cn_change}
  Given an uncertain graph $\mathcal{G} = (\mathcal{V}, \mathcal{E}, \mathcal{P})$, a possible world $G_i \sqsubseteq \mathcal{G}$ at the $i$-th iteration, an edge $e_{uv} \in \mathcal{E}$, and a node $w \in C_{\mathcal{G}}(u, v)$, where $C_{\mathcal{G}}(u,v)$ denotes the set of common neighbors of $u$ and $v$ within the topological structure of $\mathcal{G}$, ignoring edge probabilities. When the existence state of $e_{uv}$ changes in $G_i$, the number of common neighbors $|C_{G_i}(u, w)|$ is affected if and only if $e_{vw} \in E_i$.
\end{myLemma}

\begin{proof}
  $e_{uv}, e_{vw} \in E$ is the condition for $v$ to be a common neighbor of $u$ and $w$ in $G$, i.e., $v \in C_{G}(u, w)$. Let $G_{i+1}$ denote the possible world obtained by changing the existence state of $e_{uv}$ in $G_i$. We follow the existence state of $e_{vw}$ in $E_i$ to prove the lemma.

  (1) Suppose $e_{vw} \in E_i$. We next analyze the effect of the change in $e_{uv}$ state on $|C_{G_i}(u, w)|$. If $e_{uv} \in E_i$, then $e_{uv}, e_{vw} \in E_i$, so $v \in C_{G_i}(u, w)$. After removing $e_{uv}$, we have $v \notin C_{G_{i+1}}(u, w)$, so $|C_{G_{i+1}}(u, w)| = |C_{G_i}(u, w)| - 1$. If $e_{uv} \notin E_i$, then $v \notin C_{G_i}(u, w)$. After adding $e_{uv}$, $e_{uv}, e_{vw} \in E_{i+1}$, so $v \in C_{G_{i+1}}(u, w)$, giving $|C_{G_{i+1}}(u, w)| = |C_{G_i}(u, w)| + 1$. In both sub-cases, $|C_{G_i}(u, w)|$ is affected.

  (2) Suppose $e_{vw} \notin E_i$. Regardless of whether $e_{uv}$ exists or not, $v$ cannot be a common neighbor of $u$ and $w$ in either $G_i$ or $G_{i+1}$. Therefore, $|C_{G_i}(u, w)|$ is unaffected by the state change of $e_{uv}$.

  Combining both cases, $|C_{G_i}(u, w)|$ is affected by the state change of $e_{uv}$ if and only if $e_{vw} \in E_i$.
\end{proof}

\begin{myTheorem}
	\label{th:eis}
	Given an uncertain graph $\mathcal{G}=(\mathcal{V}, \mathcal{E}, \mathcal{P})$, an edge $e_{uv} \in \mathcal{E}$, a possible world $G_i \sqsubseteq \mathcal{G}$ at the $i$-th iteration (with an edge set $E_i$), the edge influence set of $e_{uv}$ in $G_i$, denoted as $S_i(e_{uv})$, is calculated by Eq. \ref{eq:eis}.
  \begin{equation}
		\label{eq:eis}
		S_i(e_{uv}) = \{e_{uw} : e_{vw} \in E_i, w \in C_{\mathcal{G}}(u,v)\}
	\end{equation}
\end{myTheorem}

\begin{proof}
  Consider the edge $e_{uw}$ with $w \in C_{\mathcal{G}}(u, v)$ and the state change of $e_{uv}$ in $G_i$. By Definition~\ref{def:tcn_dis}, we have the edge discrepancy of $e_{uw}$ depends on $|C_{G_i}(u, w)|$. By Lemma~\ref{lem:cn_change}, when the state of $e_{uv}$ changes, $|C_{G_i}(u, w)|$ is affected if and only if $e_{vw} \in E_i$. Therefore, the edge discrepancy of $e_{uw}$ is affected by the state change of $e_{uv}$ if and only if $e_{vw} \in E_i$, which yields Eq.~\ref{eq:eis}.
\end{proof}

\begin{figure}[t]
  \setlength{\abovecaptionskip}{0.1cm}
    \centering
    \includegraphics[width=0.9\linewidth]{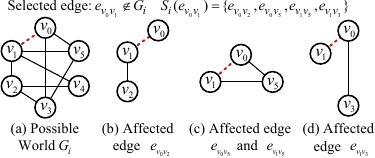}
    \caption{Identifying the edge influence set $S_i(e_{v_0v_1})$ of the selected edge $e_{v_0v_1}$ in the possible world $G_i$.}
    \label{fig:influence_set_example}
    \vspace{-0.4cm}
\end{figure}

\begin{myExample}
\label{exp:influence_set}
We use Figure~\ref{fig:influence_set_example} to illustrate how the edge influence set is identified. Figure~\ref{fig:influence_set_example}(a) shows a possible world $G_i$ of the uncertain graph $\mathcal{G}$ in Figure~\ref{fig:intro_example}(a), where the selected edge $e_{v_0v_1}$ (red dashed) is absent. In the topology of $\mathcal{G}$, $v_0$ and $v_1$ have three common neighbors: $C_{\mathcal{G}}(v_0,v_1)=\{v_2,v_5,v_3\}$. We examine them in turn. For $v_2$ in Figure~\ref{fig:influence_set_example}(b), the edge $e_{v_1v_2}$ is present in $G_i$, so changing the state of $e_{v_0v_1}$ alters the number of common neighbors of $v_0$ and $v_2$, and thus the edge discrepancy of $e_{v_0v_2}$. We therefore add $e_{v_0v_2}$ to $S_i(e_{v_0v_1})$. For $v_5$ in Figure~\ref{fig:influence_set_example}(c), both $e_{v_0v_5}$ and $e_{v_1v_5}$ are present in $G_i$, so the state change of $e_{v_0v_1}$ affects the edge discrepancies of $e_{v_0v_5}$ and $e_{v_1v_5}$, both of which are added to $S_i(e_{v_0v_1})$. For $v_3$ in Figure~\ref{fig:influence_set_example}(d), the edge $e_{v_0v_3}$ is present in $G_i$, which adds $e_{v_1v_3}$. Collecting these edges yields $S_i(e_{v_0v_1})=\{e_{v_0v_2},\, e_{v_0v_5},\, e_{v_1v_5},\, e_{v_1v_3}\}$.
\end{myExample}

Having identified $S_i(e_{uv})$ for an edge $e_{uv}$, we next characterize how the edge discrepancy of each $e_{xy} \in S_i(e_{uv})$ changes when the state of $e_{uv}$ is altered and finally present an optimized computation method of the total discrepancy variation.

\begin{myTheorem}
  \label{th:dis_change}
  Given an uncertain graph $\mathcal{G} = (\mathcal{V}, \mathcal{E}, \mathcal{P})$, a possible world $G_i \sqsubseteq \mathcal{G}$, an edge $e_{uv} \in \mathcal{E}$ with edge influence set $S_i(e_{uv})$, and $G_{i+1}$ denoting the possible world obtained by changing the existence state of $e_{uv}$ in $G_i$. For any edge $e_{xy} \in S_i(e_{uv})$, the edge discrepancy $dis_{G_{i+1}}(e_{xy})$ is given by Eq.~\ref{eq:dis_change}.
  \begin{equation}
    \label{eq:dis_change}
    dis_{G_{i+1}}(e_{xy}) =
    \begin{cases}
      dis_{G_i}(e_{xy}) - 1, & \text{if } e_{uv} \in E_i, \\
      dis_{G_i}(e_{xy}) + 1, & \text{if } e_{uv} \notin E_i.
    \end{cases}
  \end{equation}
\end{myTheorem}

\begin{proof}
  By Theorem~\ref{th:eis}, $e_{xy} \in S_i(e_{uv})$ implies there exists $w \in C_{\mathcal{G}}(u, v)$ such that $e_{xy} = e_{uw}$ and $e_{vw} \in E_i$. By Lemma~\ref{lem:cn_change}, since $e_{vw} \in E_i$, the state change of $e_{uv}$ affects $|C_{G_i}(u, w)|$: $|C_{G_{i+1}}(u, w)| = |C_{G_i}(u, w)| - 1$ if $e_{uv} \in E_i$, and $|C_{G_{i+1}}(u, w)| = |C_{G_i}(u, w)| + 1$ if $e_{uv} \notin E_i$. Since $dis_{G_i}(e_{xy}) = |C_{G_i}(u, w)| - \overline{C}_{\mathcal{G}}(u, w)$ and $\overline{C}_{\mathcal{G}}(u, w)$ is independent of the state of $e_{uv}$, $dis_{G_{i+1}}(e_{xy})$ changes by the same amount as $|C_{G_{i+1}}(u, w)|$, yielding Eq.~\ref{eq:dis_change}.
\end{proof}

\begin{myTheorem}
  \label{th:tcn_dis_var_opt1}
  Given an uncertain graph $\mathcal{G} = (\mathcal{V}, \mathcal{E}, \mathcal{P})$, a possible world $G_i \sqsubseteq \mathcal{G}$, an edge $e_{uv} \in \mathcal{E}$ with $S_i(e_{uv})$, and $G_{i+1}$ denoting the possible world obtained by changing the existence state of $e_{uv}$ in $G_i$, the total discrepancy variation $\delta_{G_i}(e_{uv})$ is given by Eq.~\ref{eq:tcn_dis_var_opt1}, where $dis_{G_{i+1}}(e_{xy})$ for each $e_{xy} \in S_i(e_{uv})$ is determined by Theorem~\ref{th:dis_change}.
  \begin{equation}
    \small
    \label{eq:tcn_dis_var_opt1}
    \delta_{G_i}(e_{uv}) = \sum_{e_{xy} \in S_i(e_{uv})}
    \left( |dis_{G_{i+1}}(e_{xy})| - |dis_{G_i}(e_{xy})| \right)
  \end{equation}
\end{myTheorem}

\begin{proof}
  By Definition~\ref{def:tcn_dis}, $\delta_{G_i}(e_{uv}) = dis_{\mathcal{G}}(G_{i+1}) - dis_{\mathcal{G}}(G_i) = \sum_{e_{xy} \in \mathcal{E}} (|dis_{G_{i+1}}(e_{xy})| - |dis_{G_i}(e_{xy})|)$. By Theorem~\ref{th:eis}, only edges in $S_i(e_{uv})$ have their edge discrepancies affected by the state change of $e_{uv}$, so the terms for edges outside $S_i(e_{uv})$ vanish, yielding Eq.~\ref{eq:tcn_dis_var_opt1}. For each remaining term $e_{xy} \in S_i(e_{uv})$, Theorem~\ref{th:dis_change} gives the exact value of $dis_{G_{i+1}}(e_{xy})$, which directly determines $|dis_{G_{i+1}}(e_{xy})| - |dis_{G_i}(e_{xy})|$.
\end{proof}

\begin{algorithm}[t]
  \small
  \setstretch{0.85}
  \caption{\texttt{Random Selection-based Refinement}}
  \label{alg:rsr}
  \LinesNumbered
  \KwIn{uncertain graph $\mathcal{G}=(\mathcal{V}, \mathcal{E}, \mathcal{P})$, iterative starting point $G_{0} = (V_{0},E_{0})$, iterations $k\geq 0$}
  \KwOut{$G_{C}$}

  \For{$i \leftarrow$ 1 to $k$}{
    $G_{i} \leftarrow G_{i-1}$\;
    $e_{uv}$ $\leftarrow$ randomly select an edge from $\mathcal{E}$\;
    $S_i(e_{uv})$ $\leftarrow$ get the influence set of $e_{uv}$ \tcc*{Theorem \ref{th:eis}}
    $\delta_{G_i}(e_{uv})$ $\leftarrow$ compute total discrepancy variation via $S_i(e_{uv})$ \tcc*{Theorem \ref{th:tcn_dis_var_opt1}}
    \If{$\delta_{G_i}(e_{uv}) < 0$}{
      \If{$e_{uv} \in E_{i}$}{
        remove $e_{uv}$ from $E_{i}$\;
      }
      \Else{
        add $e_{uv}$ to $E_{i}$\;
      }
    }
    \Else{
      \continue
    }
  }

  \Return $G_{C} = G_{k}$\;
\end{algorithm}

\noindent\textbf{Algorithm of Random Selection-based Refinement}. Next, we introduce the procedure of the random selection-based refinement in Algorithm~\ref{alg:rsr}. In each iteration, starting from the current possible world $G_i$ ($i\in [1,k]$) obtained from the previous iteration (line 2), we randomly select an edge $e_{uv} \in \mathcal{E}$ and identify its edge influence set $S_i(e_{uv})$ (lines 3-4). Then, for each edge $e_{xy} \in S_i(e_{uv})$, we compute $\overline{C}_{\mathcal{G}}(x, y)$, and substitute $\overline{C}_{\mathcal{G}}(x, y)$ into Eq. \ref{eq:tcn_dis_var_opt1} to obtain $\delta_{G_i}(e_{uv})$ by Theorem~\ref{th:tcn_dis_var_opt1} (line 5). If $\delta_{G_i}(e_{uv}) < 0$, we change the existence state of $e_{uv}$ in $G_i$ to obtain $G_{i+1}$ (lines 6-10); otherwise, we reject the change and keep $G_{i+1} = G_i$ (lines 11-12). After $k$ iterations, we return the final possible world $G_k$ (line 13).

\vspace{0.1cm}
\noindent\textbf{Complexity Analysis}.
We derive the time complexity of the basic algorithm by analyzing its two stages, i.e., the probability-based initialization (Algorithm~\ref{alg:pin}) and the random selection-based refinement (Algorithm~\ref{alg:rsr}). In the initialization stage, Algorithm~\ref{alg:pin} performs a random number generation and a probability comparison for each edge in $\mathcal{E}$, each taking $O(1)$ time. Thus, this stage takes $O(|\mathcal{E}|)$ time. The random selection-based refinement consists of $k$ iterations. The cost of each iteration is divided into three parts: (1) randomly selecting an edge, which takes $O(1)$ time; (2) computing the total discrepancy variation $\delta_{G_i}(e_{uv})$ via Eq.~\ref{eq:tcn_dis_var_opt1}, whose cost is dominated by identifying the edge influence set $S_i(e_{uv})$. The cost of identifying $S_i(e_{uv})$ is in turn dominated by computing the common neighbors $C_{\mathcal{G}}(u,v)$, which is obtained by intersecting the neighbor sets of the two endpoints $u$ and $v$. Let $\bar{d}$ denote the average degree of nodes in $\mathcal{G}$. We use the hash set to store the neighbor set of each node. This intersection takes $O(\bar{d})$ time, and hence the second operation takes $O(\bar{d})$ time; and (3) changing the state of the selected edge, which takes $O(1)$ time. Therefore, each iteration takes $O(\bar{d})$ time, and the refinement stage takes $O(k\, \bar{d})$ time. Combining the two stages, the overall time complexity of the basic algorithm is $O\big(|\mathcal{E}| + k\, \bar{d}\big)$.

\section{Efficient Counting-Based Refinement}
\label{sec:integer_based_refinement}

The core of Algorithm~\ref{alg:rsr} is to decide whether changing the state of selected edge $e_{uv}$ reduces total discrepancy, i.e., whether $\delta_{G_i}(e_{uv})<0$. However, computing the exact variation $\delta_{G_i}(e_{uv})$ requires numerous calculations at each iteration, which include computing the updated discrepancies of affected edges (Theorem 4), calculating the changes in their absolute values, and summing these changes over $S_i(e_{uv})$ (Theorem 5). On Flickr dataset, these calculations account for 63.2\% of  total runtime. However, the decision does not actually require exact value of $\delta_{G_i}(e_{uv})$; determining the trend of total discrepancy variation, i.e., whether $\delta_{G_i}(e_{uv}) \geqslant 0$ or $< 0$, is sufficient. Therfore, we propose an efficient counting-based estimation to replace tardy exact computation. With each edge discrepancy rounded to the nearest integer, we derive a counting criterion to determine this trend. Experiments in (\S\ref{sec:experiments}) show substantial efficiency gains with limited quality loss.

\noindent\textbf{Overview}. We show how rounding enables us to determine whether $\delta_{G_i}(e_{uv}) < 0$ through counting. By Eq.~\ref{eq:tcn_dis_var_opt1}, $\delta_{G_i}(e_{uv})$ sums $|dis_{G_{i+1}}(e_{xy})| - |dis_{G_i}(e_{xy})|$ over all $e_{xy} \in S_i(e_{uv})$. After rounding, each edge discrepancy is an integer and changes by $+1$ or $-1$ under a state change of $e_{uv}$ (Theorem~\ref{th:dis_change}); hence each $|dis_{G_{i+1}}(e_{xy})| - |dis_{G_i}(e_{xy})|$ is $+1$ or $-1$. Thus, $\delta_{G_i}(e_{uv})<0$ exactly when the $-1$ contributions outnumber the $+1$ contributions, allowing the determination to be made throughcounting. We next present the \textit{counting-based determination}, and corresponding \textit{algorithm of counting-based refinement}.

\vspace{0.1cm}
\noindent\textbf{Counting-based Determination}. For each $e_{xy} \in S_i(e_{uv})$, whether its term $|dis_{G_{i+1}}(e_{xy})| - |dis_{G_i}(e_{xy})|$ is $+1$ or $-1$ depends on the sign of $dis_{G_i}(e_{xy})$. We therefore partition $S_i(e_{uv})$ into three subsets $S_i^{<0}(e_{uv})$, $S_i^{=0}(e_{uv})$, and $S_i^{>0}(e_{uv})$, corresponding to $dis_{G_i}(e_{xy}) < 0$, $dis_{G_i}(e_{xy}) = 0$, and $dis_{G_i}(e_{xy}) > 0$, and analyze the $|dis_{G_{i+1}}(e_{xy})| - |dis_{G_i}(e_{xy})|$ within each subset in Lemmas~\ref{lem:dis_{G_i}(e_{xy})<0}-\ref{lem:dis_{G_i}(e_{xy})=0}.

\begin{myLemma}
	\label{lem:dis_{G_i}(e_{xy})<0}
	Given an uncertain graph $\mathcal{G} = (\mathcal{V}, \mathcal{E}, \mathcal{P})$, a possible world $G_i \sqsubseteq \mathcal{G}$ at the $i$-th iteration, an $e_{uv} \in \mathcal{E}$ with influence set $S_i(e_{uv})$ and an edge $e_{xy} \in S_i^{<0}(e_{uv})$. Changing the state of $e_{uv}$ from absent to present decreases the absolute edge discrepancy of $e_{xy}$ by 1; otherwise, it increases by 1.
\end{myLemma}

\begin{proof}
  \label{prf:dis_{G_i}(e_{xy})<0}
  Let $G_i$ and $G_{i+1}$ denote the possible worlds before and after the state change of $e_{uv}$. According to Theorem~\ref{th:dis_change}, when the existence state of $e_{uv}$ changes from absent to present, then the edge $e_{xy} \in S_i^{<0}(e_{uv})$, which satisfies $dis_{G_i}(e_{xy}) < 0$, has $dis_{G_i}(e_{xy}) + 1$ and the variation of the absolute edge discrepancy of $e_{xy}$ is given as follows:

  \begin{align*}
  \left| dis_{G_{i+1}}(e_{xy}) \right| - \left| dis_{G_i}(e_{xy}) \right|
  &= \left| dis_{G_i}(e_{xy}) + 1 \right| - \left| dis_{G_i}(e_{xy}) \right| \\
  &= -\left[ dis_{G_i}(e_{xy}) + 1 \right] + dis_{G_i}(e_{xy}) \\
  &= -1 \quad .
  \end{align*}
  
  Conversely, when the existence state of $e_{uv}$ changes from present to absent, then the edge $e_{xy} \in S_i^{<0}(e_{uv})$ has $dis_{G_i}(e_{xy}) - 1$ and the variation of the absolute edge discrepancy is given as follows:

  \begin{align*}
  \left| dis_{G_{i+1}}(e_{xy}) \right| - \left| dis_{G_i}(e_{xy}) \right|
  &= \left| dis_{G_i}(e_{xy}) - 1 \right| - \left| dis_{G_i}(e_{xy}) \right| \\
  &= -\left[ dis_{G_i}(e_{xy}) - 1 \right] + dis_{G_i}(e_{xy}) \\
  &= 1 \quad .
  \end{align*}

  Therefore, the lemma holds.
\end{proof}

With the lemma \ref{lem:dis_{G_i}(e_{xy})<0}, for any edge $e_{xy} \in S_i^{<0}(e_{uv})$, the variation of its absolute edge discrepancy under the existence state change of $e_{uv}$ can be obtained efficiently.

\begin{myLemma}
	\label{lem:dis_{G_i}(e_{xy})>0}
	Given an uncertain graph $\mathcal{G} = (\mathcal{V}, \mathcal{E}, \mathcal{P})$, a possible world $G_i \sqsubseteq \mathcal{G}$ at the $i$-th iteration, an $e_{uv} \in \mathcal{E}$ with influence set $S_i(e_{uv})$ and an edge $e_{xy} \in S_i^{>0}(e_{uv})$. Changing the state of $e_{uv}$ from absent to present increases the absolute edge discrepancy of $e_{xy}$ by 1; otherwise, it decreases by 1.
\end{myLemma}

\begin{proof}
  \label{prf:dis_{G_i}(e_{xy})>0}
  Let $G_i$ and $G_{i+1}$ denote the possible worlds before and after the state change of $e_{uv}$, respectively. Similar to the proof of Lemma \ref{lem:dis_{G_i}(e_{xy})<0}, when the existence state of $e_{uv}$ changes from absent to present, for the edge $e_{xy} \in S_i^{>0}(e_{uv})$ satisfying $dis_{G_i}(e_{xy}) > 0$, the variation of its absolute edge discrepancy is given as follows:

  \begin{align*}
  \left| dis_{G_{i+1}}(e_{xy}) \right| - \left| dis_{G_i}(e_{xy}) \right|
  &= \left| dis_{G_i}(e_{xy}) + 1 \right| - \left| dis_{G_i}(e_{xy}) \right| \\
  &= \left[ dis_{G_i}(e_{xy}) + 1 \right] - dis_{G_i}(e_{xy}) \\
  &= 1 \quad .
  \end{align*}
  
  Similarly, when the existence state of $e_{uv}$ changes from present to absent, the variation of the absolute edge discrepancy of $e_{xy}$ is given as follows:

  \begin{align*}
  \left| dis_{G_{i+1}}(e_{xy}) \right| - \left| dis_{G_i}(e_{xy}) \right|
  &= \left| dis_{G_i}(e_{xy}) - 1 \right| - \left| dis_{G_i}(e_{xy}) \right| \\
  &= \left[ dis_{G_i}(e_{xy}) - 1 \right] - dis_{G_i}(e_{xy}) \\
  &= -1 \quad .
  \end{align*}

  Therefore, the lemma holds.
\end{proof}

Leveraging Lemma~\ref{lem:dis_{G_i}(e_{xy})>0}, for any edge $e_{xy} \in S_i^{>0}(e_{uv})$, the variation of its absolute edge discrepancy under the existence state change of $e_{uv}$ can be efficiently determined.

\begin{myLemma}
	\label{lem:dis_{G_i}(e_{xy})=0}
	Given an uncertain graph $\mathcal{G} = (\mathcal{V}, \mathcal{E}, \mathcal{P})$, a possible world $G_i \sqsubseteq \mathcal{G}$ at the $i$-th iteration, an $e_{uv} \in \mathcal{E}$ with influence set $S_i(e_{uv})$ and an edge $e_{xy} \in S_i^{=0}(e_{uv})$. Regardless of the existence state change of $e_{uv}$, the variation of the absolute edge discrepancy of $e_{xy}$ is always increase by 1.
\end{myLemma}

\begin{proof}
  \label{prf:dis_{G_i}(e_{xy})=0}
  Let $G_i$ and $G_{i+1}$ denote the possible worlds before and after the state change of $e_{uv}$, respectively. The variation of the absolute edge discrepancy is derived as follows:

  \begin{align*}
  \left| dis_{G_{i+1}}(e_{xy}) \right| - \left| dis_{G_i}(e_{xy}) \right|
  &= \left| dis_{G_i}(e_{xy}) \pm  1 \right| - \left| dis_{G_i}(e_{xy}) \right| \\
  &= \left| 0 \pm 1 \right| - \left| 0 \right| \\
  &= 1 \quad .
  \end{align*}
\end{proof}

Building on Lemma~\ref{lem:dis_{G_i}(e_{xy})=0}, for any edge $e_{xy} \in S_i^{=0}(e_{uv})$, the variation of its absolute edge discrepancy can be readily obtained if the existence state of $e_{uv}$ changes.

Next, we utilize the three subsets $S_i^{<0}(e_{uv})$, $S_i^{=0}(e_{uv})$, and $S_i^{>0}(e_{uv})$ to analyze the condition under which $\delta_{G_i}(e_{uv}) < 0$, as stated in the following theorem.

\begin{myTheorem}
	\label{th:tcn_dis_opt2}
  Given an uncertain graph $\mathcal{G} = (\mathcal{V}, \mathcal{E}, \mathcal{P})$, a possible world $G_i \sqsubseteq \mathcal{G}$ at the $i$-th iteration, and an edge $e_{uv} \in \mathcal{E}$. If $e_{uv}$ is absent (present) in $G_i$ and $|S_i^{<0}(e_{uv})| > |S_i^{=0}(e_{uv})| + |S_i^{>0}(e_{uv})|$ ($ |S_i^{>0}(e_{uv})| > |S_i^{=0}(e_{uv})| + |S_i^{<0}(e_{uv})|$), then the total discrepancy variation $\delta_{G_i \sqsubseteq \mathcal{G}}(e_{uv}) < 0$.
\end{myTheorem}

\begin{proof}
  \label{prf:tcn_dis_opt2}
  We prove Theorem~\ref{th:tcn_dis_opt2} according to the existence state of $e_{uv}$ in $G_i$. (1) If $e_{uv} \notin E_i$, according to Lemmas~\ref{lem:dis_{G_i}(e_{xy})<0}-\ref{lem:dis_{G_i}(e_{xy})=0}, when the state of $e_{uv}$ changes (i.e., $e_{uv} \in E_{i+1}$), the absolute edge discrepancies of edges in $S_i^{<0}(e_{uv})$ decrease by 1, while those in $S_i^{=0}(e_{uv})$ and $S_i^{>0}(e_{uv})$ increase by 1. For $\delta_{G_i}(e_{uv}) < 0$ to hold, it must satisfy that $|S_i^{<0}(e_{uv})| > |S_i^{=0}(e_{uv})| + |S_i^{>0}(e_{uv})|$. (2) Conversely, if $e_{uv} \in E_i$, when the state of $e_{uv}$ changes (i.e., $e_{uv} \notin E_{i+1}$), the absolute edge discrepancies of edges in $S_i^{>0}(e_{uv})$ decrease by 1, while those in $S_i^{=0}(e_{uv})$ and $S_i^{<0}(e_{uv})$ increase by 1. For $\delta_{G_i}(e_{uv}) < 0$ to hold, it must satisfy that $|S_i^{>0}(e_{uv})| > |S_i^{=0}(e_{uv})| + |S_i^{<0}(e_{uv})|$.
\end{proof}

With Theorem~\ref{th:tcn_dis_opt2}, we can efficiently determine whether $\delta_{G_i}(e_{uv}) < 0$ when the existence state of $e_{uv}$ is changed, which in turn determines whether to update $e_{uv}$ and obtain $G_{i+1}$.

\vspace{0.1cm}
\noindent\textbf{Algorithm of Counting-based Refinement}. Algorithm~\ref{alg:ir} details the counting-based refinement. Given an uncertain graph $\mathcal{G}$ and an initial possible world $G_0$ produced by Algorithm~\ref{alg:pin}, it first rounds every $dis_{G_0}(e)$ to the nearest integer (lines 1-2). Then, in the $i$-th iteration, it sets $G_i=G_{i-1}$, randomly selects an edge $e_{uv}$, and constructs $S_i(e_{uv})$ (lines 3-6). It then partitions $S_i(e_{uv})$ into $S_i^{<0}(e_{uv})$, $S_i^{=0}(e_{uv})$, and $S_i^{>0}(e_{uv})$ and counts their cardinalities according to Lemmas~\ref{lem:dis_{G_i}(e_{xy})<0}-\ref{lem:dis_{G_i}(e_{xy})=0} (line 7). By Theorem~\ref{th:tcn_dis_opt2}: if $e_{uv} \notin E_i$ and $|S_i^{<0}(e_{uv})| > |S_i^{=0}(e_{uv})| + |S_i^{>0}(e_{uv})|$, the algorithm adds $e_{uv}$ to $E_i$ and increments $dis_{G_i}(e_{xy})$ by 1 for each $e_{xy} \in S_i(e_{uv})$ (lines 8-10). Conversely, if $e_{uv} \in E_i$ and $|S_i^{>0}(e_{uv})| > |S_i^{=0}(e_{uv})| + |S_i^{<0}(e_{uv})|$, it removes $e_{uv}$ from $E_i$ and decrements $dis_{G_i}(e_{xy})$ by 1 for each $e_{xy} \in S_i(e_{uv})$ (lines 11-13). Maintaining the integer-valued $dis_{G_i}(e_{xy})$ through the $\pm 1$ updates in lines 10 and 14 ensures the validity of the counting-based determination in subsequent iterations. After $k$ iterations, the algorithm returns $G_C = G_k$ as the final possible world (line 14).

\vspace{0.1cm}
\noindent\textbf{Complexity Analysis}.
We analyze the time complexity of the counting-based refinement (Algorithm~\ref{alg:ir}). Let $\bar{d}$ denote the average degree of nodes in $\mathcal{G}$, and the neighbor set of each node is stored as a hash set. The algorithm consists of two parts: the integer rounding (lines 1-2) and the iterative refinement (lines 3-13). In the rounding stage, the algorithm rounds $dis_{G_0}(e)$ to the nearest integer for each edge $e \in \mathcal{E}$, taking $O(|\mathcal{E}|)$ time. The refinement stage consists of $k$ iterations, each performing four steps. First, it randomly selects an edge $e_{uv}$ in $O(1)$ time. Second, it identifies the edge influence set $S_i(e_{uv})$, whose cost is dominated by intersecting the neighbor sets of $u$ and $v$ to compute the common neighbors $C_{\mathcal{G}}(u,v)$, taking $O(\bar{d})$ time. Third, it determines whether to change the state of $e_{uv}$ by counting $|S_i^{<0}(e_{uv})|$, $|S_i^{=0}(e_{uv})|$, and $|S_i^{>0}(e_{uv})|$, which scans $S_i(e_{uv})$ once. Since $|S_i(e_{uv})| \leq |C_{\mathcal{G}}(u,v)| $, this step takes $O(\bar{d})$ time. Fourth, if the state of $e_{uv}$ is changed, it updates $dis_{G_i}(e_{xy})$ for each $e_{xy} \in S_i(e_{uv})$, taking $O(\bar{d})$ time. Hence, each iteration takes $O(\bar{d})$ time, and the refinement stage takes $O(k\, \bar{d})$ time. Combining two stages, the counting-based refinement takes $O(|\mathcal{E}| + k\, \bar{d})$ time. It does not lower the asymptotic complexity of random selection-based refinement, but it replaces the tardy exact computation in each iteration with counting, whose efficiency gain is demonstrated in our experiments (\S\ref{sec:experiments}).

\begin{algorithm}[t]
  \small
  \setstretch{0.85}
  \caption{\texttt{Counting-based Refinement}}
  \label{alg:ir}
  \LinesNumbered
  \KwIn{uncertain graph $\mathcal{G}=(\mathcal{V}, \mathcal{E}, \mathcal{P})$, initial possible world $G_{0} = (V_{0}, E_{0})$, iterations $k \geq 0$}
  \KwOut{$G_{C}$}

  \For{$e_{uv} \in \mathcal{E}$}{
    $dis_{G_0}(e_{uv}) \leftarrow$ round $dis_{G_0}(e_{uv})$ to the nearest integer\;
  }
  \For{$i \leftarrow 1$ to $k$}{
    $G_i \leftarrow G_{i-1}$\;
    $e_{uv} \leftarrow$ randomly select an edge from $\mathcal{E}$\;
    $S_i(e_{uv}) \leftarrow$ get the edge influence set of $e_{uv}$ \tcc*{Theorem~\ref{th:eis}}
    count $|S_i^{<0}(e_{uv})|$, $|S_i^{=0}(e_{uv})|$, $|S_i^{>0}(e_{uv})|$ \tcc*{Lemma~\ref{lem:dis_{G_i}(e_{xy})<0}-\ref{lem:dis_{G_i}(e_{xy})=0}}
    \tcp{{\footnotesize determine whether to change $e_{uv}$ by Theorem~\ref{th:tcn_dis_opt2}}}
    \If{$|S_i^{<0}(e_{uv})| > |S_i^{=0}(e_{uv})| + |S_i^{>0}(e_{uv})|$ and $e_{uv} \notin E_i$}{
      add $e_{uv}$ to $E_i$\;
      $dis_{G_i}(e_{xy}) \leftarrow dis_{G_i}(e_{xy}) + 1$ for all $e_{xy} \in S_i(e_{uv})$\;
    }
    \If{$|S_i^{>0}(e_{uv})| > |S_i^{=0}(e_{uv})| + |S_i^{<0}(e_{uv})|$ and $e_{uv} \in E_i$}{
      remove $e_{uv}$ from $E_i$\;
      $dis_{G_i}(e_{xy}) \leftarrow dis_{G_i}(e_{xy}) - 1$ for all $e_{xy} \in S_i(e_{uv})$\;
    }
  }

  \Return $G_{C} = G_k$\;
\end{algorithm}

\section{Beta-based Adaptive Termination}
\label{sec:beta_based_termination}

In counting-based refinement (\S\ref{sec:integer_based_refinement}), users preset  iteration $k$ to obtain a possible world of the desired quality. However, selecting an appropriate $k$ is challenging because the number of iterations needed to reach the quality target is unknown in advance. An excessive iteration budget incurs unnecessary computation after the target is reached, whereas an insufficient one leads to premature termination. This motivates us to let users specify the desired quality directly and assess the current possible world during refinement to determine when to stop.

To measure refinement quality of the current possible world $G_i$, we observe that refinement reduces $dis_{\mathcal{G}}(G_i)$ via edge state changes, so the proportion of edges whose state changes can further reduce this discrepancy reflects refinement quality. We define an edge $e_{uv} \in \mathcal{E}$ as \textit{state-changeable} under $G_i$ if changing its state reduces $dis_{\mathcal{G}}(G_i)$, and take the proportion of such edges in $\mathcal{E}$ as the \textit{state-changeable edge proportion} (SEP). As refinement proceeds, fewer state-changeable edges remain and SEP decreases. We introduce a user-specified \textit{quality threshold} $\tau_q \in (0, 1)$: once $\mathrm{SEP}<\tau_q$, $G_i$ is considered sufficiently refined and refinement terminates.

However, computing the SEP exactly requires enumerating every edge in $\mathcal{E}$, which is infeasible. We therefore estimate the SEP through edge sampling. Since the Beta distribution is well suited to modeling an unknown proportion, we adopt it to quantify the confidence that the SEP is below $\tau_q$. The user further specifies a \textit{confidence threshold} $\tau_c \in (0, 1)$, the minimum confidence they require, and the refinement terminates once this confidence reaches $\tau_c$.

In the following, we model the estimation of the SEP via the Beta distribution in \S\ref{sec:beta_modeling}. Built on this model, we establish the termination criterion for the counting-based refinement in \S\ref{sec:beta_criterion} and  present the complete algorithm in \S\ref{sec:beta_algorithm}.

\subsection{Beta-based Estimation Modeling}
\label{sec:beta_modeling}

The Beta distribution, denoted by $\mathrm{Beta}(\alpha, \beta)$, is a continuous probability distribution supported on $[0, 1]$, with shape parameters $\alpha, \beta > 0$. Specifically, given a series of independent Bernoulli trials with $\alpha - 1$ successes and $\beta - 1$ failures, $\mathrm{Beta}(\alpha, \beta)$ models the distribution of the unknown success probability. The probability density function $f(x; \alpha, \beta)$ of $\mathrm{Beta}(\alpha, \beta)$ is given by Eq.~\ref{eq:beta_pdf}, where $B(\alpha, \beta)$ denotes the Beta function \cite{abramowitz1948handbook}. For fixed parameters $\alpha$ and $\beta$, the value of $f(x; \alpha, \beta)$ at a given $x$ indicates how credible it is that the true success probability takes the value $x$.

\begin{equation}
f(x; \alpha, \beta) = \frac{x^{\alpha-1}(1-x)^{\beta-1}}{B(\alpha, \beta)}, \quad x \in [0, 1]
\label{eq:beta_pdf}
\end{equation}

\vspace{0.1cm}
\noindent\textbf{Overview}. We model the SEP estimation via $\mathrm{Beta}(\alpha, \beta)$, where the parameters $\alpha$ and $\beta$ are determined by the number of successes and failures through edge sampling, respectively. Therefore, in the following, we first describe the process of edge sampling; and then model the SEP estimation via $\mathrm{Beta}(\alpha, \beta)$.

\vspace{0.1cm}
\noindent\textbf{Process of Edge Sampling.} To estimate the SEP, we repeatedly sample $n$ edges uniformly at random from $\mathcal{E}$, where $n$ denotes the total number of edge samples. For each sampled edge, we record a success if it is state-changeable and a failure otherwise, constituting a Bernoulli trial whose success probability equals the true SEP. After $n$ samples, we denote the number of successes as $m$ and the number of failures as $n - m$. These values are used to determine the parameters of the Beta distribution $\mathrm{Beta}(\alpha, \beta)$.

\vspace{0.1cm}
\noindent\textbf{Beta Distribution Modeling of SEP Estimation.} Given $m$ state-changeable edges ($m$ successes) and $n - m$ non-state-changeable edges ($n - m$ failures) from the edge sampling, we model the estimation of SEP via $\mathrm{Beta}(\alpha, \beta)$, where $\alpha = m+1$ and $\beta = n-m+1$. Accordingly, the probability density function of $\mathrm{Beta}(\alpha, \beta)$ can be expressed as $f(x; m+1, n-m+1)$, which reflects how credible it is that the true SEP takes the value $x$ with $m$ and $n-m$.

\subsection{Termination Criterion}
\label{sec:beta_criterion}

\vspace{0.1cm}
\noindent\textbf{Overview}. Built on the Beta model in \S\ref{sec:beta_modeling}, which characterizes how credible each possible value of the true SEP is, the confidence that SEP $< \tau_q$ can be quantified. We therefore first formalize the termination criterion for the counting-based refinement. Since the confidence that SEP $< \tau_q$ is computed via $f(x; \alpha, \beta)$ whose parameters depend on the total number of edge samples $n$, we then determine the required sample size $n$ to make the criterion operational.

\vspace{0.1cm}
\noindent\textbf{Termination Criterion.} Since $f(x; \alpha, \beta)$ reflects how credible it is that the true SEP takes the value $x$, the confidence that SEP $< \tau_q$ is obtained by integrating $f(x; m+1, n-m+1)$ over $[0, \tau_q]$, where $n$ is the total number of edge samples and $m$ is the number of state-changeable edges observed among them. The termination criterion is formalized as follows.

\begin{myTheorem}
\label{th:beta_termination}
Given an uncertain graph $\mathcal{G} = (\mathcal{V}, \mathcal{E}, \mathcal{P})$, a possible world $G_i \sqsubseteq \mathcal{G}$ at the $i$-th iteration, the total number of edge samples $n$, the number of state-changeable edges observed $m$, a quality threshold $\tau_q \in (0,1)$, and a confidence threshold $\tau_c \in (0,1)$, the counting-based refinement terminates if and only if Eq.~\ref{eq:termination_criterion} holds.
\begin{equation}
\int_0^{\tau_q} f(x; m+1, n-m+1) \, dx \geq \tau_c
\label{eq:termination_criterion}
\end{equation}
\end{myTheorem}

\begin{proof}
Once the confidence that SEP $< \tau_q$ reaches $\tau_c$, the current possible world $G_i$ has been refined to the degree the user expects, and the counting-based refinement can be terminated. The confidence is given by the integral of $f(x; m+1, n-m+1)$ over $[0, \tau_q]$, i.e., $\int_0^{\tau_q} f(x; m+1, n-m+1) \, dx$. Therefore, the termination condition is satisfied if and only if this integral is at least $\tau_c$.
\end{proof}

The termination criterion in Theorem~\ref{th:beta_termination} depends on both $m$ and $n$. While $m$ is observed dynamically during edge sampling, $n$ must be determined before the sampling begins. We therefore determine the required sample size $n$ as follows.

\vspace{0.1cm}
\noindent\textbf{Sample-Size Determination.}
The criterion in Eq.~\ref{eq:termination_criterion} depends on the sample size $n$ and the observed success count $m$. Increasing $n$ strengthens the statistical evidence for SEP estimation but also raises the sampling cost. To balance them, we require the confidence that the true SEP is below $\tau_q$ to reach $\tau_c$ while minimize the sampling cost. We denote the minimum sample size by $n^*$. Given $\tau_q$ and $\tau_c$, a smaller $m$ require a smaller $n*$ to satisfy Eq.~\ref{eq:termination_criterion} \cite{hornik2024two}. Although $m=0$ minimizes $n^*$, it means a zero-success sample and provides no tolerance for an occasional state-changeable edge. Therefore, we set $m=1$ as a practical design choice and use the corresponding minimum sample size $n^*$ as the number of sampling iterations. The following theorem guarantees the existence of such an $n^*$.

\begin{myTheorem}
\label{th:min_sample_size}
For any $\tau_q, \tau_c \in (0, 1)$, there exists a minimum positive integer $n$, i.e., $n^*$, such that Eq. \ref{eq:min_sample_size} is satisfied.
\begin{equation}
\int_0^{\tau_q} f(x; 2, n^*) \, dx \geq \tau_c
\label{eq:min_sample_size}
\end{equation}
\end{myTheorem}

\begin{proof}
Since the integral is monotonically increasing as $n$ increases \cite{hornik2024two}, the integral over $[0, \tau_q]$ increases monotonically toward $1 > \tau_c$ as $n \to \infty$. Therefore, there exists a $n^*$ such that the integral first reaches $\tau_c$.
\end{proof}

With the existence of $n^*$ established, we next determine its value. Since the integral $\int_0^{\tau_q} f(x; 2, n^*-1)\,dx$ cannot be analytically inverted with respect to $n^*$, so $n^*$ must be determined via numerical search. To determine the value of $n^*$ more quickly, we locate $n^*$ by exponential search: starting from an initial value $n_0 = \lceil 1/\tau_q \rceil$, we repeatedly double $n_0$ until the integral first exceeds $\tau_c$, requiring only $O(\log n^*)$ evaluations of the integral.

\subsection{Algorithm of beta-based Refinement}
\label{sec:beta_algorithm}

We present the algorithm of Beta-based refinement in Algorithm~\ref{alg:beta_integer}. It takes an uncertain graph $\mathcal{G}$ and an initial possible world $G_0$ as input, and replaces the fixed iteration count $k$ in Algorithm~\ref{alg:ir} with a quality threshold $\tau_q$ and a confidence threshold $\tau_c$, enabling adaptive termination of the counting-based refinement based on the criterion in Theorem~\ref{th:beta_termination}.

\vspace{0.1cm}
Algorithm~\ref{alg:beta_integer} begins by rounding $dis_{G_0}(e_{uv})$ to the nearest integer for each edge $e_{uv} \in \mathcal{E}$ (lines 1-2). It then determines $n^*$, the minimum sample size satisfying Eq.~\ref{eq:min_sample_size}, via exponential search (line 3). The state-changeable edge counter $m$ and the iteration counter $i$ are initialized to $0$ and $1$, respectively (line 4).
In each iteration, the algorithm initializes $G_i$ from $G_{i-1}$ (line 6), randomly selects an edge $e_{uv} \in \mathcal{E}$, and determines whether to change its state in $G_i$ via Theorem~\ref{th:tcn_dis_opt2} (lines 7-17), following the same procedure as lines 5-13 in Algorithm~\ref{alg:ir}. If the state of $e_{uv}$ is changed, $e_{uv}$ is state-changeable under $G_{i}$, and $m$ is incremented by 1 (lines 13, 17).
After each $n^*$ iterations, the termination criterion in Theorem~\ref{th:beta_termination} is evaluated (lines 19-20). If the criterion is satisfied, the counting-based refinement terminates with $G_i$ (lines 20-21); otherwise, $m$ is reset to $0$ and the counting-based refinement proceeds to the next iteration (line 22). Once the counting-based refinement terminates, it outputs $G_C = G_i$.

\vspace{0.1cm}
\noindent\textbf{Remark.} Strictly speaking, the termination criterion in Theorem~\ref{th:beta_termination} assumes that the $n^*$ samples are drawn from a fixed possible world. However, for the counting-based refinement in Algorithm~\ref{alg:beta_integer}, each iteration serves simultaneously as one edge sample, so the $n^*$ samples within a window span a sequence of possible worlds $G_i, G_{i+1}, \ldots, G_{i+n^*-1}$. This design avoids the overhead of a separate sampling procedure: rather than explicitly sampling $n^*$ edges from $\mathcal{E}$ at each checkpoint, the algorithm reuses the edges already selected during the counting-based refinement iterations, incurring no additional cost. This deviation from the theoretical assumption remains acceptable. Since each state change in the counting-based refinement tends to reduce the SEP, the possible worlds in the early part of the window, i.e., $G_i, G_{i+1}$, tend to have a higher SEP than $G_{i+n^*-1}$ at the end of the window. Consequently, the observed $m$ tends to be larger than what would be obtained by sampling exclusively from $G_{i+n^*-1}$. As the integral in Eq.~\ref{eq:termination_criterion} decreases monotonically with $m$ \cite{hornik2024two}, an overestimated $m$ makes the criterion harder to satisfy, i.e., more conservative. Therefore, once the criterion is satisfied, the true SEP of $G_{i+n^*-1}$ is reliably at or below $\tau_q$.

\begin{algorithm}[t]
  \small
  \setstretch{0.85}
  \caption{\texttt{Beta-based Refinement}}
  \label{alg:beta_integer}
  \LinesNumbered
  \KwIn{uncertain graph $\mathcal{G}=(\mathcal{V}, \mathcal{E}, \mathcal{P})$, initial possible world $G_{0} = (V_{0}, E_{0})$, quality threshold $\tau_q \in (0,1)$, confidence threshold $\tau_c \in (0,1)$}
  \KwOut{$G_{C}$}

  \For{$e_{uv} \in \mathcal{E}$}{
    $dis_{G_0}(e_{uv}) \leftarrow$ round $dis_{G_0}(e_{uv})$ to the nearest integer\;
  }
  $n^* \leftarrow$ minimum sample size satisfying Eq.~\ref{eq:min_sample_size} \tcc*{Theorem~\ref{th:min_sample_size}}  
  $m \leftarrow 0, i \leftarrow 1$\;

  \While{true}{
    $G_i \leftarrow G_{i-1}$\;
    $e_{uv} \leftarrow$ randomly select an edge from $\mathcal{E}$\;
    $S_i(e_{uv}) \leftarrow$ get the edge influence set of $e_{uv}$ \tcc*{Theorem~\ref{th:eis}}
    count $|S_i^{<0}(e_{uv})|$, $|S_i^{=0}(e_{uv})|$, $|S_i^{>0}(e_{uv})|$ \tcc*{Lemma~\ref{lem:dis_{G_i}(e_{xy})<0}-\ref{lem:dis_{G_i}(e_{xy})=0}}
    \tcp{{\footnotesize determine whether to change $e_{uv}$ by Theorem~\ref{th:tcn_dis_opt2}}}
    \If{$|S_i^{<0}(e_{uv})| > |S_i^{=0}(e_{uv})| + |S_i^{>0}(e_{uv})|$ and $e_{uv} \notin E_i$}{
      add $e_{uv}$ to $E_i$\;
      $dis_{G_i}(e_{xy}) \leftarrow dis_{G_i}(e_{xy}) + 1$ for all $e_{xy} \in S_i(e_{uv})$\;
      $m \leftarrow m + 1$\;
    }
    \If{$|S_i^{>0}(e_{uv})| > |S_i^{=0}(e_{uv})| + |S_i^{<0}(e_{uv})|$ and $e_{uv} \in E_i$}{
      remove $e_{uv}$ from $E_i$\;
      $dis_{G_i}(e_{xy}) \leftarrow dis_{G_i}(e_{xy}) - 1$ for all $e_{xy} \in S_i(e_{uv})$\;
      $m \leftarrow m + 1$\;
    }
    \If{$i \bmod n^* = 0$}{
      \If{$\int_0^{\tau_q} f(x;\, m+1,\, n^*-m+1)\,dx \geq \tau_c$}{
        \textbf{break} \tcc*{Theorem~\ref{th:beta_termination}}}
      $m \leftarrow 0$\;
    }
    $i \leftarrow i + 1$\;
  }

  \Return $G_{C} = G_i$\;
\end{algorithm}

\section{Experiments}
\label{sec:experiments}

Our experiments address the following questions. All methods are implemented in Python 3.12 and run on a Linux server with a 3.7\,GHz CPU and 128\,GB memory. Our code is available at \cite{code}.

\smallskip
\noindent\textbf{Q1:} How do our methods compare with existing RPW methods in preserving structural features?
(\S\ref{sec:exp_effectiveness})

\noindent\textbf{Q2:} How close are our results to the exact CRPW?
(\S\ref{sec:Comparison with the Exact CRPW})

\noindent\textbf{Q3:} How efficient are the proposed algorithms?
(\S\ref{sec:exp_efficiency})

\noindent\textbf{Q4:} Is the Beta-based termination effective in controlling
the quality of the resulting representative possible world?
(\S\ref{sec:exp_beta})

\noindent\textbf{Q5:} How does practical sampling compare with theory-aligned sampling in SEP estimation and runtime? 
(\S\ref{sec:Practical vs. Theory-Aligned Sampling})

\noindent\textbf{Q6:} How sensitive are our methods to parameters?
(\S\ref{sec:exp_sensitivity})

\begin{table}[t]
  \small
  \caption{Statistics of datasets}
    \label{tab:datasets}
    \vspace{-0.4cm}
  \resizebox{\columnwidth}{!}{\begin{tabular}{c||c|c|cc}
    \multirow{2}{*}{\textbf{Datasets} $\downarrow$} & \multirow{2}{*}{\textbf{Nodes}} & \multirow{2}{*}{\textbf{Edges}} & \multicolumn{2}{c}{\boldmath$\overline{C}_{\mathcal{G}}(u,v)$}       \\ \cline{4-5} 
                              &                        &                        & \multicolumn{1}{c|}{\textbf{mean $\pm$ SD}}      & \textbf{\{25\%, 50\%, 75\%\}}       \\ \hline \hline
    Rath                     & 11                 & 32                & \multicolumn{1}{c|}{0.72 $\pm$ 0.43} & \{0.40, 0.65, 1.06\} \\ \hline
    Email                     & 36,692                 & 183,831                & \multicolumn{1}{c|}{2.98 $\pm$ 4.21} & \{0.63, 1.55, 3.68\} \\ \hline
    Flickr                    & 720,504                & 16,003,787             & \multicolumn{1}{c|}{1.36 $\pm$ 2.96} & \{0.04, 0.22, 1.13\} \\ \hline
    Biomine                   & 1,863,844              & 34,047,306             & \multicolumn{1}{c|}{4.05 $\pm$ 7.61} & \{0.40, 1.42, 4.08\}  \\ \hline
    DBLP                      & 3,255,283              & 17,165,252             & \multicolumn{1}{c|}{3.03 $\pm$ 3.40} & \{1.26, 2.29, 3.62\} \\ 
  \end{tabular}}
  \vspace{-0.5cm}
\end{table}

\subsection{Experimental Setup} 
\label{sec:exp_setup}

\noindent\textbf{Datasets.}
Table~\ref{tab:datasets} summarizes the statistics of five real-world uncertain graphs spanning diverse application domains and a wide range of graph sizes.
(1) Rath \cite{diviak2019structure} is a small corruption network. The edge probabilities are generated from a uniform distribution. We use it to obtain the exact optimum by brute-force enumeration, thereby enabling us to quantify the optimality gaps of our methods.
(2) Email \cite{klimt2004enron} is an email communication network. Each node is an email address, and an undirected edge connects two addresses if at least one email was sent between them. The edge probabilities are generated from a uniform distribution.
(3) Flickr \cite{potamias2010knn} is a photo-sharing social network. Each node is a user, and an edge connects two users who share at least oneinterest group. The probability of an edge is set to the Jaccard coefficient of the interest groups of its two endpoints.
(4) Biomine \cite{eronen2012biomine} is a protein interaction network. Each node is a protein, and an undirected edge connects two interacting proteins. The probability of an edge reflects the confidence that the corresponding interaction exists.
(5) DBLP \cite{ley2002dblp, potamias2010knn} is an academic collaboration network. Each node is an author, and an edge connects two authors who have coauthored at least one publication. The probability of an edge is derived from an exponential function of the number of collaborations between its two endpoints. 
The last two columns in Table~\ref{tab:datasets} summarize the distribution of $\overline{C}_{\mathcal{G}}(u,v)$ over all edges $(u,v) \in \mathcal{E}$, reporting its mean, standard deviation together with its 25th, 50th, and 75th percentiles.

\noindent\textbf{Methods.}
We compare our three methods against three baselines.
\noindent\underline{Our methods.} (1) P+RSR is our basic algorithm, consisting of two stages: probability-based initialization (P, \S\ref{sec:basic_algorithm}), which produces an initial possible world $G_0$, and random selection-based refinement (RSR), which reduces the total discrepancy of $G_0$. (2) P+CR optimizes the efficiency of RSR with counting-based refinement (CR, \S\ref{sec:integer_based_refinement}), which replaces the exact computation in each iteration with counting. (3) P+BCR builds upon CR with Beta-based refinement (BCR, \S\ref{sec:beta_based_termination}), which adaptively terminates the refinement instead of running a fixed number of iterations.

\noindent\underline{Baselines.}
(4) Exact enumerates all possible worlds to solve the CRPW problem. (5) ADR (Average Degree Rewiring) and (6) ABM (Approximate B-Matching) \cite{parchas2014pursuit} construct a degree-based representative possible world that preserves the expected node degrees. (7) TRPW \cite{song2016triangle} constructs a triangle-based representative possible world that preserves both the expected node degrees and the expected triangle degrees.

\noindent\textbf{Evaluation Dimensions.}
Our evaluation is organized along two dimensions. \textit{Representativeness} measures how well the structural features computed on a resulting possible world approximate their expected values over the uncertain graph. \textit{Refinement quality} measures how far our methods optimize the possible world.

\noindent\textbf{Metrics.} 
For \textit{Representativeness}, we adopt five structural features as the metrics, and organize them into three groups by their relationship to the number of common neighbors, the feature our methods optimize. The first group, \textit{common neighbor-defined features}, consists of features whose definition contains this number directly. The second group, \textit{common neighbor-derived features}, consists of features whose definition does not contain this number directly but are instead governed by intermediate structures that themselves derive from common neighbors. The third group, \textit{global features}, consists of global features that bear no direct algebraic relationship to common neighbors. This organization lets us examine whether the benefit of optimizing common neighbors extends to features derived from them and, further, to features unrelated to them.

\noindent\underline{Common Neighbor-defined Features.} This group contains two features. (1) \textbf{The number of common neighbors} (NCN) of a node pair $(u,v)$ is $|N(u) \cap N(v)|$, the feature our methods directly optimize. (2) \textbf{The Jaccard coefficient} of $(u,v)$ is $|N(u) \cap N(v)| / |N(u) \cup N(v)|$ \cite{kogge2016jaccard}, whose numerator is the number of common neighbors. 

\noindent\underline{Common Neighbor-derived Features.} 
This group contains two features. (3) \textbf{The trussness} of an edge, from the $k$-truss decomposition \cite{cohen2008trusses}, is the largest $k$ such that the edge belongs to a subgraph where every edge lies in at least $k-2$ triangles. The number of triangles on an edge equals the number of common neighbors of its two endpoints, so the trussness derives from common neighbors. (4) \textbf{The clustering coefficient} of a node $v$ is $\frac{2 E_v}{d_v(d_v-1)}$ \cite{watts1998collective}, where $E_v$ is the number of edges among its neighbors and $d_v$ is its degree; it depends on triangles and thus is derived from common neighbors.

\noindent\underline{Global Features.} This group contains (5) \textbf{The shortest-path distance} (SP distance). The feature is measured over pairs of reachable nodes, which is a  global feature with no direct algebraic relationship to the number of common neighbors.

For \textit{Refinement Quality}, we measure the refinement quality of a possible world by its mean absolute discrepancy (MAD). For a possible world $G \sqsubseteq \mathcal{G}$, MAD is the total discrepancy of $G$ (Definition~\ref{def:tcn_dis}) divided by the number of edges, given by $\mathrm{MAD}(G) = \frac{dis_{\mathcal{G}}(G)}{|\mathcal{E}|}$. A smaller MAD indicates higher refinement quality.

\noindent\textbf{Default Parameters.}
Unless otherwise specified, each method uses the following settings. For P+RSR and P+CR, we use $k = |\mathcal{E}|$ iterations to balance refinement quality and running time. For P+BCR, the quality threshold and the confidence threshold are set to $\tau_q = 0.02$ and $\tau_c = 0.9$. ADR \cite{parchas2014pursuit} and TRPW \cite{song2016triangle} run $10 \sum_{e \in \mathcal{E}} p(e)$ iterations, following the settings in their papers. ABM \cite{parchas2014pursuit} has no iteration count to set.

\subsection{Effectiveness Evaluation}
\label{sec:exp_effectiveness}

We evaluate each method by the residual between the feature distribution of the possible world it produces and the \textit{expected} feature distribution over the uncertain graph $\mathcal{G}$, where \textit{expected} denotes the feature distribution that $\mathcal{G}$ induces in expectation. The former distribution is obtained directly: each method outputs a single deterministic possible world, on which we compute the feature value of every node (or node pair, or edge) and tally these values into a distribution. For example, for the common neighbor count, we count, over all node pairs in the output possible world, how many pairs share $0$ common neighbors, how many share $1$, and so on. The \textit{expected} distribution is obtained as follows. For the common neighbor count, it can be computed exactly, since the expected count of a node pair equals a sum of products of edge probabilities (Definition~\ref{def:exp_cn}); we tally these expected counts into a distribution in the same way. The remaining features cannot compute expectation directly, so we estimate their \textit{expected} distributions by Monte Carlo over $1000$ possible worlds sampled from $\mathcal{G}$ \cite{parchas2014pursuit, potamias2010knn}, averaging the resulting distributions. We then partition the range of each feature into five equal-width bins and sum the counts falling in each bin, yielding a five-bin distribution for both the method and \textit{expected}. For the $j$-th bin, let $c_j$ and $\hat{c}_j$ be the summed counts in the distribution of a method and in \textit{expected}, respectively. The residual is $r_j = \left| \log_{10}\!\left( \frac{c_j + 1}{\hat{c}_j + 1} \right) \right|$, so $r_j = 0$ indicates exact agreement with \textit{expected}, and a lower residual curve over the five bins indicates a closer match. For the shortest-path distance, all-pairs computation is too costly on the full graph; following \cite{parchas2014pursuit, song2016triangle}, both the methods and the estimation are run on a forest-fire subgraph \cite{leskovec2006sampling} instead.

\noindent\textbf{Common Neighbor-defined Features.}
This group comprises the number of common neighbors and the Jaccard coefficient.

Figure~\ref{fig:cn} and Table~\ref{tab:cn} report the results of the number of common neighbors. Table~\ref{tab:cn} lists the mean $E(r)$ and variance $var(r)$ of the residual over each dataset for all methods, where the three shades of red, from dark to light, mark the lowest, second-lowest, and third-lowest values in each row; the same convention applies to all subsequent tables. Across all four datasets, the residual curves of our three methods stay close to $r_j = 0$, and their $E(r)$ in Table~\ref{tab:cn} is consistently the lowest, followed by ADR and TRPW, with ABM the worst. This confirms that directly optimizing the number of common neighbors yields a possible world whose common neighbor distribution faithfully reflects $\overline{C}_{\mathcal{G}}(u,v)$. Averaged over all datasets, P+CR and P+BCR stay within $0.12$ of P+RSR, showing that the counting-based refinement preserves representativeness while improving efficiency.

Figure~\ref{fig:jaccard} and Table~\ref{tab:jaccard} report the results of the jaccard coefficient. Our three methods again stay close to \textit{Expected}, and their $E(r)$ in Table~\ref{tab:jaccard} is the lowest on most datasets. The counting-based refinements P+CR and P+BCR trail P+RSR by a small margin, and on some datasets even surpass it, confirming that they preserve representativeness while improving efficiency.

\begin{figure}
  \setlength{\abovecaptionskip}{0.1cm}
    \centering
    \includegraphics[width=0.95\linewidth]{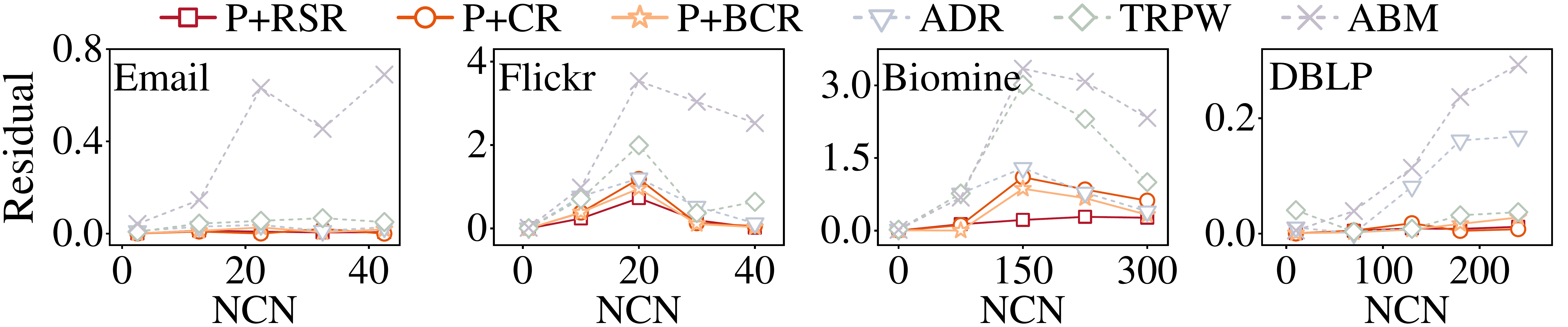}
    \vspace{-0.1cm}
    \caption{Residual of NCN.}
    \label{fig:cn}
    \vspace{-0.4cm}
\end{figure}

\begin{table}[t]

  \small
  \caption{$E(r)$ and $var(r)$ of the residual on NCN}
  \label{tab:cn}
  \resizebox{\columnwidth}{!}{
    \begin{tabular}{cc||c|c|c|c|c|c}
      \multicolumn{2}{c||}{\textbf{Methods} $\rightarrow$}                           & \textbf{P+RSR} & \textbf{P+CR} & \textbf{P+BCR} & \textbf{ADR} & \textbf{TRPW} & \textbf{ABM} \\ \hline \hline
      \multicolumn{1}{c|}{\multirow{2}{*}{Email}}   & E(r)   & \cellcolor{ranktwo} \SI{6.62E-03}{}     &  \cellcolor{rankone} \SI{6.10E-03}{}    & \cellcolor{rankthree} \SI{1.27E-02}{}     & \SI{2.36E-02}{}   & \SI{4.50E-02}{}    & \SI{3.93E-01}{}   \\ \cline{2-8} 
      \multicolumn{1}{c|}{}                         & var(r) & \cellcolor{rankone} \SI{1.95E-5}{}     & \cellcolor{ranktwo} \SI{7.85E-05}{}    & \cellcolor{rankthree} \SI{1.04e-4}{}     & \SI{1.36e-4}{}  & \SI{4.94e-4}{}   & \SI{8.32e-2}{}  \\ \hline
      \multicolumn{1}{c|}{\multirow{2}{*}{Flickr}}  & E(r)   & \cellcolor{rankone} \SI{2.36E-01}{}     & \cellcolor{rankthree} \SI{3.50E-01}{}    & \cellcolor{ranktwo} \SI{2.96E-01}{}     & \SI{5.25E-01}{}   & \SI{7.43E-01}{}    & \SI{2.02}{}   \\ \cline{2-8} 
      \multicolumn{1}{c|}{}                         & var(r) & \cellcolor{rankone} \SI{8.63E-02}{}     & \cellcolor{rankthree}\SI{2.36E-01}{}    & \cellcolor{ranktwo} \SI{1.58E-01}{}     & \SI{2.39E-01}{}  & \SI{5.67E-01}{}   & \SI{2.17}{}  \\ \hline
      \multicolumn{1}{c|}{\multirow{2}{*}{Biomine}} & E(r)   & \cellcolor{rankone} \SI{1.81e-1}{}     & \cellcolor{rankthree} \SI{5.37e-1}{}    & \cellcolor{ranktwo} \SI{3.74e-1}{}     & \SI{6.53e-1}{}   & \SI{1.422}{} & \SI{1.890}{}   \\ \cline{2-8} 
      \multicolumn{1}{c|}{}                         & var(r) & \cellcolor{rankone} \SI{1.36e-2}{}     & \cellcolor{rankthree}\SI{2.23e-1}{}    & \cellcolor{ranktwo} \SI{1.51e-1}{}     & \SI{2.27e-1}{}  & \SI{1.47}{}   & \SI{2.17}{}  \\ \hline
      \multicolumn{1}{c|}{\multirow{2}{*}{DBLP}}    & E(r)   & \cellcolor{rankone} \SI{6.88E-03}{}     & \cellcolor{ranktwo} \SI{7.16E-03}{}    & \cellcolor{rankthree} \SI{1.10E-02}{}     & \SI{8.45E-02}{}   &  \SI{2.42E-02}{}    &  \SI{1.38E-01}{}   \\ \cline{2-8} 
      \multicolumn{1}{c|}{}                         & var(r) & \cellcolor{rankone} \SI{1.87e-5}{}     & \cellcolor{ranktwo}\SI{4.37e-5}{}    & \cellcolor{rankthree} \SI{1.29e-4}{}     & \SI{6.38e-3}{}  & \SI{2.98e-4}{}   & \SI{1.54e-2}{}
    \end{tabular}
  }
\end{table}

\begin{figure}
  \setlength{\abovecaptionskip}{0.1cm}
    \centering
    \includegraphics[width=0.95\linewidth]{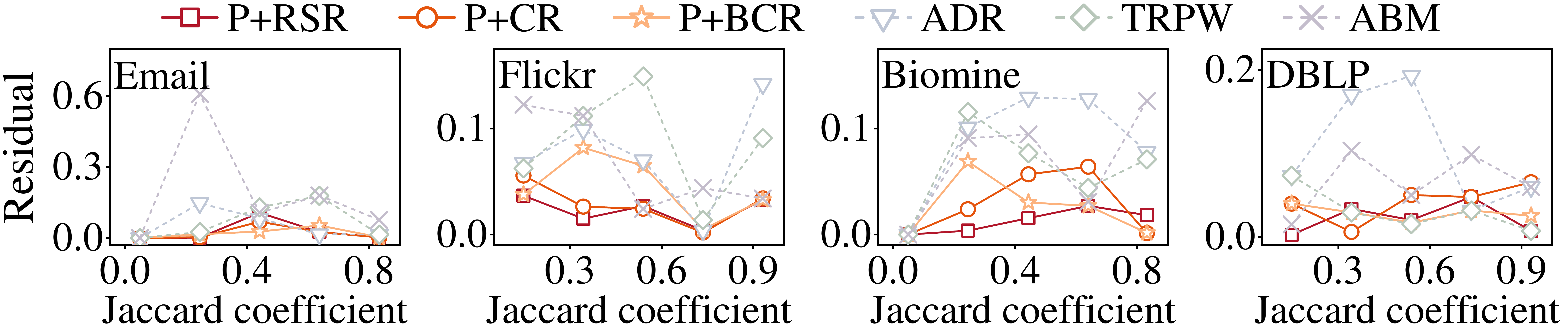}
    \vspace{-0.1cm}
    \caption{Residual of the Jaccard coefficient.}
    \label{fig:jaccard}
    \vspace{-0.4cm}
\end{figure}

\begin{table}[!t]

  \small
  \caption{$E(r)$ and $var(r)$ of the residual on the Jacc.}
  \label{tab:jaccard}
  \resizebox{\columnwidth}{!}{
    \begin{tabular}{cc||c|c|c|c|c|c}
      \multicolumn{2}{c||}{\textbf{Methods} $\rightarrow$}                           & \textbf{P+RSR} & \textbf{P+CR} & \textbf{P+BCR} & \textbf{ADR} & \textbf{TRPW} & \textbf{ABM} \\ \hline \hline
      \multicolumn{1}{c|}{\multirow{2}{*}{Email}}   & E(r)   & \cellcolor{rankthree} \SI{2.75E-2}{}     & \cellcolor{ranktwo} \SI{2.03E-2}{}     & \cellcolor{rankone} \SI{2.02E-2}{}      & \SI{5.22E-2}{}    & \SI{6.97E-2}{}     & \SI{1.95E-1}{}    \\ \cline{2-8} 
      \multicolumn{1}{c|}{}                         & var(r) & \cellcolor{rankthree} \SI{2.09E-3}{}     & \cellcolor{ranktwo} \SI{8.61E-4}{}     & \cellcolor{rankone} \SI{4.66E-4}{}      & \SI{3.89E-3}{}    & \SI{6.34E-3}{}     & \SI{5.80E-2}{}  \\ \hline
      \multicolumn{1}{c|}{\multirow{2}{*}{Flickr}}  & E(r)   & \cellcolor{rankone} \SI{2.31E-2}{}     & \cellcolor{ranktwo} \SI{2.85E-2}{}     & \cellcolor{rankthree} \SI{4.45E-2}{}      &  \SI{7.64E-2}{}    & \SI{8.57E-2}{}     & \SI{6.74E-2}{}    \\ \cline{2-8} 
      \multicolumn{1}{c|}{}                         & var(r) & \cellcolor{rankone} \SI{1.82E-4}{}     & \cellcolor{ranktwo} \SI{3.76E-4}{}     & \cellcolor{rankthree} \SI{8.93E-4}{}      &  \SI{2.55E-3}{}    & \SI{2.61E-3}{}     & \SI{2.14E-3}{}  \\ \hline
      \multicolumn{1}{c|}{\multirow{2}{*}{Biomine}} & E(r)   & \cellcolor{rankone} \SI{1.28E-2}{}     & \cellcolor{rankthree} \SI{2.92E-2}{}     &  \cellcolor{ranktwo} \SI{2.56E-2}{}      & \SI{8.72E-2}{}    & \SI{6.14E-2}{}     & \SI{6.85E-2}{}   \\ \cline{2-8} 
      \multicolumn{1}{c|}{}                         & var(r) & \cellcolor{rankone} \SI{1.22E-4}{}     & \cellcolor{rankthree} \SI{9.09E-4}{}     & \cellcolor{ranktwo} \SI{7.83E-4}{}      & \SI{2.82E-3}{}    & \SI{1.84E-3}{}     & \SI{2.65E-3}{}  \\ \hline
      \multicolumn{1}{c|}{\multirow{2}{*}{DBLP}}    & E(r)   & \cellcolor{rankone} \SI{2.22E-2}{}     & \cellcolor{rankthree} \SI{3.13E-2}{}     & \cellcolor{ranktwo} \SI{2.85E-2}{}      & \SI{1.06E-1}{}    & \SI{3.55E-2}{}     & \SI{6.53E-2}{}   \\ \cline{2-8} 
      \multicolumn{1}{c|}{}                         & var(r) & \cellcolor{ranktwo} \SI{3.47E-4}{}     & \cellcolor{rankthree} \SI{4.88E-4}{}     & \cellcolor{rankone} \SI{7.04E-5}{}      & \SI{5.11E-3}{}    & \SI{6.49E-4}{}     & \SI{1.34E-3}{}  \\ 
    \end{tabular}
  }
  
\end{table}

\noindent\textbf{Common Neighbor-derived Features.}
This group comprises the trussness and the clustering coefficient.

Figure~\ref{fig:trussness} and Table~\ref{tab:Trussness} report the results of the trussness. In Figure~\ref{fig:trussness}, the residual curves of our three methods and TRPW all stay close to $r_j = 0$, far below those of ADR and ABM. Moreover, in terms of the statistics in Table~\ref{tab:Trussness}, our three methods slightly outperform TRPW: averaging $E(r)$ across the four datasets, P+RSR, P+CR, and P+BCR attain $0.045$, $0.075$, and $0.054$, all below the $0.084$ of TRPW.

Figure~\ref{fig:clustering} and Table~\ref{tab:Clustering} report the results of the clustering coefficient. The results mirror those for trussness: in Figure~\ref{fig:clustering}, the residual curves of our three methods and TRPW far below those of ADR and ABM. While the average $E(r)$ across the four datasets of our methods slightly outperform TRPW: P+RSR, P+CR, and P+BCR attain $0.036$, $0.046$, and $0.044$, all below the $0.049$ of TRPW.

These results show that preserving the number of common neighbors helps keep the features derived from it close to \textit{Expected}.

\begin{figure}
  \setlength{\abovecaptionskip}{0.1cm}
    \centering
    \includegraphics[width=0.95\linewidth]{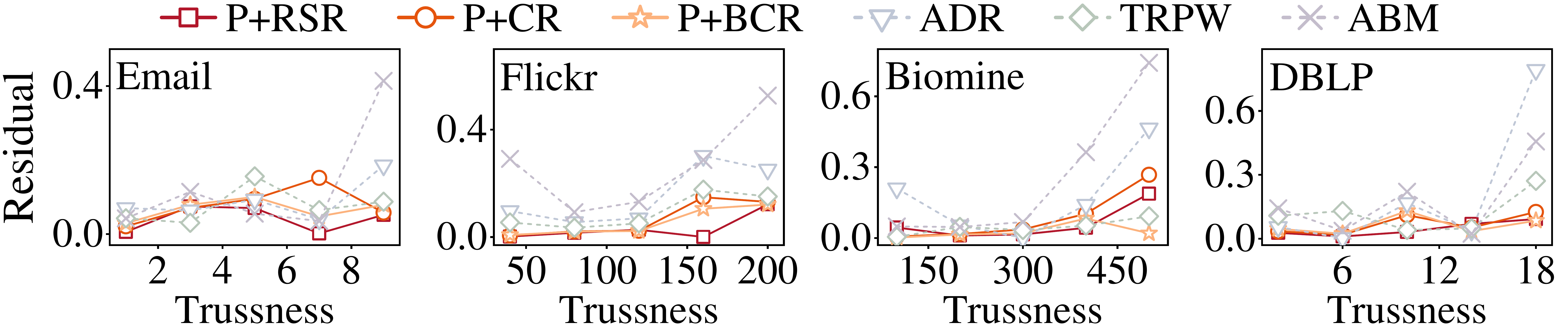}
    \vspace{-0.1cm}
    \caption{Residual of the trussness.}
    \label{fig:trussness}
    \vspace{-0.4cm}
\end{figure}

\begin{table}[]

  \small
  \caption{Mean $E(r)$ and variance $var(r)$ of the residual on the trussness}
  \label{tab:Trussness}
  \vspace{-0.4cm}
  \resizebox{\columnwidth}{!}{
    \begin{tabular}{cc||c|c|c|c|c|c}
      \multicolumn{2}{c||}{\textbf{Methods} $\rightarrow$}                           & \textbf{P+RSR} & \textbf{P+CR} & \textbf{P+BCR} & \textbf{ADR} & \textbf{TRPW} & \textbf{ABM} \\ \hline \hline
      \multicolumn{1}{c|}{\multirow{2}{*}{Email}}   & E(r)   & \cellcolor{rankone} \SI{4.09E-2}{}     & \SI{7.92E-2}{}     & \cellcolor{ranktwo}\SI{6.71E-2}{}      & \SI{9.04E-2}{}    & \cellcolor{rankthree}\SI{7.63E-2}{}     & \SI{1.32E-1}{}    \\ \cline{2-8} 
      \multicolumn{1}{c|}{}                         & var(r) & \cellcolor{ranktwo} \SI{1.22E-3}{}     & \cellcolor{rankthree} \SI{2.36E-3}{}     & \cellcolor{rankone} \SI{8.06E-4}{}      & \SI{3.20E-3}{}    & \SI{2.46E-3}{}     & \SI{2.58E-2}{}  \\ \hline
      \multicolumn{1}{c|}{\multirow{2}{*}{Flickr}}  & E(r)   &  \cellcolor{rankone} \SI{3.37E-2}{}    & \cellcolor{rankthree} \SI{6.61E-2}{}     & \cellcolor{ranktwo} \SI{5.56E-2}{}      & \SI{1.55E-1}{}    & \SI{9.37E-2}{}     & \SI{2.67E-1}{}    \\ \cline{2-8} 
      \multicolumn{1}{c|}{}                         & var(r) &  \cellcolor{rankone} \SI{2.56E-3}{}   & \SI{4.61E-3}{}     & \cellcolor{ranktwo} \SI{2.85E-3}{}      & \SI{1.29E-2}{}    & \cellcolor{rankthree} \SI{4.24E-3}{}     & \SI{2.93E-2}{}  \\ \hline
      \multicolumn{1}{c|}{\multirow{2}{*}{Biomine}} & E(r)   &  \cellcolor{rankthree} \SI{6.04E-2}{}     & \SI{8.63E-2}{}     & \cellcolor{rankone} \SI{2.80E-2}{}      & \SI{1.74E-1}{}    & \cellcolor{ranktwo} \SI{4.64E-2}{}     & \SI{2.54E-1}{}   \\ \cline{2-8} 
      \multicolumn{1}{c|}{}                         & var(r) &  \cellcolor{rankthree} \SI{5.32E-3}{}     & \SI{1.17E-2}{}     & \cellcolor{ranktwo} \SI{9.94E-4}{}      & \SI{3.27E-2}{}    & \cellcolor{rankone} \SI{9.93E-4}{}     & \SI{9.25E-2}{}  \\ \hline
      \multicolumn{1}{c|}{\multirow{2}{*}{DBLP}}    & E(r)   &  \cellcolor{rankone} \SI{4.65E-2}{}  & \cellcolor{rankthree} \SI{6.87E-2}{}     & \cellcolor{ranktwo} \SI{6.41E-2}{}     & \SI{2.11E-1}{}    & \SI{1.19E-1}{}     & \SI{1.74E-1}{}   \\ \cline{2-8} 
      \multicolumn{1}{c|}{}                         & var(r) &  \cellcolor{rankone} \SI{1.12E-3}{}    & \cellcolor{rankthree} \SI{2.23E-3}{}     & \cellcolor{ranktwo} \SI{1.90E-3}{} & \SI{1.09E-1}{}    & \SI{8.67E-3}{}     & \SI{3.11E-2}{}  \\ 
    \end{tabular}
  }
  
\end{table}

\begin{figure}
  \setlength{\abovecaptionskip}{0.1cm}
    \centering
    \includegraphics[width=0.95\linewidth]{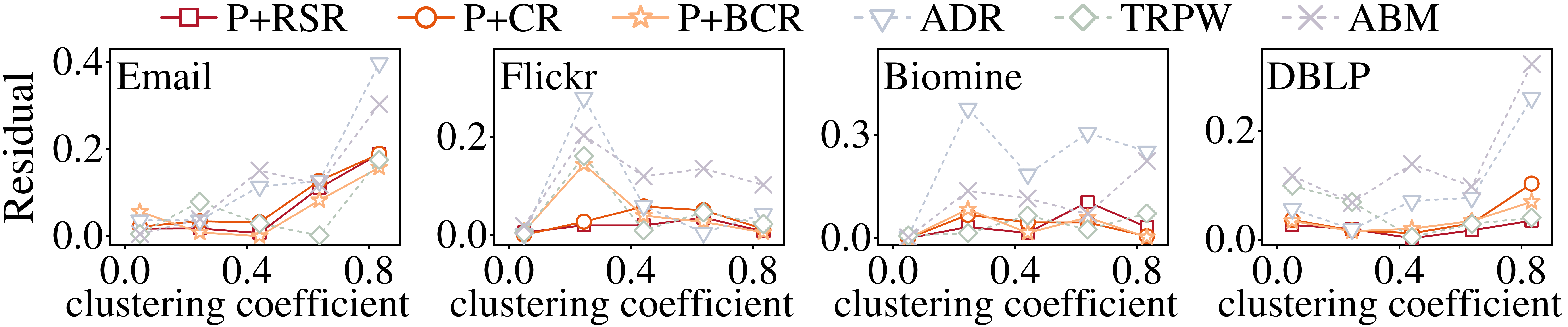}
    \vspace{-0.1cm}
    \caption{Residual of the clustering coefficient}
    \label{fig:clustering}
    \vspace{-0.4cm}
\end{figure}

\begin{table}[!t]

  \small
  \caption{$E(r)$ and $var(r)$ of the residual on the Clust.}
  \label{tab:Clustering}
  \resizebox{\columnwidth}{!}{
    \begin{tabular}{cc||c|c|c|c|c|c}
      \multicolumn{2}{c||}{\textbf{Methods} $\rightarrow$}                           & \textbf{P+RSR} & \textbf{P+CR} & \textbf{P+BCR} & \textbf{ADR} & \textbf{TRPW} & \textbf{ABM} \\ \hline \hline
      \multicolumn{1}{c|}{\multirow{2}{*}{Email}}   & E(r)   & \cellcolor{rankthree} \SI{6.90E-2}{}     & \SI{8.15E-2}{}     & \cellcolor{ranktwo} \SI{6.16E-2}{}      & \SI{1.43E-1}{}    & \cellcolor{rankone} \SI{5.86E-2}{}     & \SI{1.24E-1}{}    \\ \cline{2-8} 
      \multicolumn{1}{c|}{}                         & var(r) & \SI{6.32E-3}{}     & \cellcolor{rankthree} \SI{5.46E-3}{}     & \cellcolor{rankone} \SI{4.05E-3}{}      & \SI{2.21E-2}{}    & \cellcolor{ranktwo} \SI{5.25E-3}{}     & \SI{1.35E-2}{}  \\ \hline
      \multicolumn{1}{c|}{\multirow{2}{*}{Flickr}}  & E(r)   & \cellcolor{rankone} \SI{1.81E-2}{}     & \cellcolor{ranktwo} \SI{3.04E-2}{}     & \cellcolor{rankthree} \SI{4.52E-2}{}      & \SI{7.84E-2}{}    & \SI{4.95E-2}{}     & \SI{1.16E-1}{}    \\ \cline{2-8} 
      \multicolumn{1}{c|}{}                         & var(r) & \cellcolor{rankone} \SI{1.45E-4}{}     & \cellcolor{ranktwo} \SI{6.01E-4}{}     & \cellcolor{rankthree} \SI{3.28E-3}{}      & \SI{1.34E-2}{}    & \SI{4.17E-3}{}     & \SI{4.44E-3}{}  \\ \hline
      \multicolumn{1}{c|}{\multirow{2}{*}{Biomine}} & E(r)   & \SI{3.74E-2}{}     & \cellcolor{ranktwo} \SI{3.32E-2}{}     & \cellcolor{rankone} \SI{3.28E-2}{}      &  \SI{2.25E-1}{}    & \cellcolor{rankthree} \SI{3.72E-2}{}     & \SI{1.10E-1}{}   \\ \cline{2-8} 
      \multicolumn{1}{c|}{}                         & var(r) & \SI{1.63E-3}{}     & \cellcolor{rankone} \SI{8.44E-4}{}     & \cellcolor{rankthree} \SI{1.41E-3}{}      &  \SI{2.02E-2}{}    & \cellcolor{ranktwo} \SI{8.81E-4}{}     & \SI{6.71E-3}{}  \\ \hline
      \multicolumn{1}{c|}{\multirow{2}{*}{DBLP}}    & E(r)   & \cellcolor{rankone} \SI{2.02E-2}{}     & \cellcolor{rankthree} \SI{3.97E-2}{}     & \cellcolor{ranktwo} \SI{3.45E-2}{}      & \SI{9.68E-2}{}    & \SI{4.85E-2}{}     & \SI{1.49E-1}{}   \\ \cline{2-8} 
      \multicolumn{1}{c|}{}                         & var(r) & \cellcolor{rankone} \SI{1.43E-4}{}     & \SI{1.35E-3}{}     & \cellcolor{ranktwo} \SI{4.49E-4}{}      & \SI{8.82E-3}{}    & \cellcolor{rankthree} \SI{1.34E-3}{}     & \SI{1.01E-2}{}  \\ 
    \end{tabular}
  }
\end{table}

\noindent\textbf{Global Features.}
Figure~\ref{fig:sp} and Table~\ref{tab:ShortestPath} report the results of the shortest-path distance. In Figure~\ref{fig:sp}, it is obvious that no method keeps its advantage across all datasets, and the statistics in Table~\ref{tab:ShortestPath} confirm this: P+RSR attains the lowest $E(r)$ on Email and Biomine, TRPW on Flickr, and ABM on DBLP. Overall, although the leading method varies, our methods still achieve good results across datasets. This shows that even on the shortest-path distance, a global feature with no direct relationship to the number of common neighbors, our methods remain competitive.

\begin{figure}
  \setlength{\abovecaptionskip}{0.1cm}
    \centering
    \includegraphics[width=0.95\linewidth]{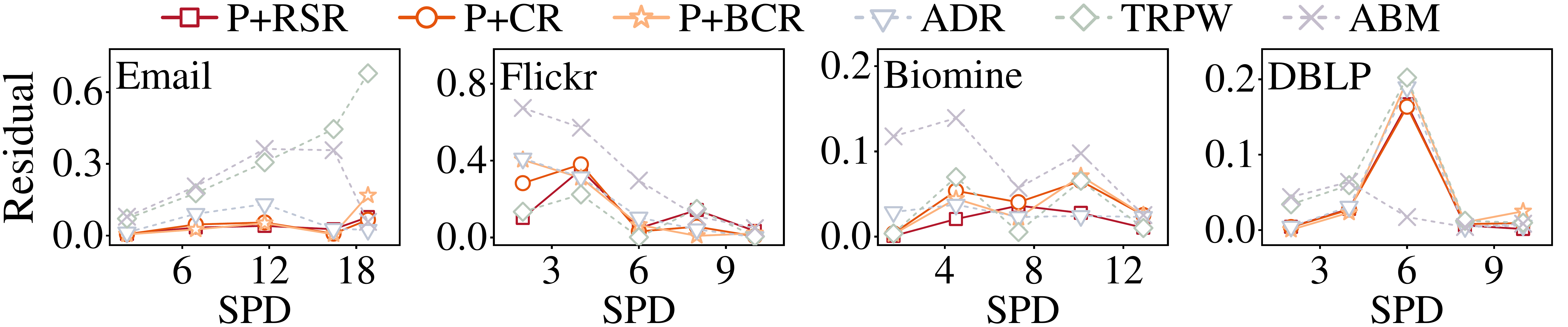}
    \vspace{-0.1cm}
    \caption{Residual of SP distance}
    \label{fig:sp}
    \vspace{-0.4cm}
\end{figure}

\begin{table}[t]

  \small
  \caption{$E(r)$ and $var(r)$ of the residual on the SPD}
  \label{tab:ShortestPath}
  \resizebox{\columnwidth}{!}{
    \begin{tabular}{cc||c|c|c|c|c|c}
      \multicolumn{2}{c||}{\textbf{Methods} $\rightarrow$}                           & \textbf{P+RSR} & \textbf{P+CR} & \textbf{P+BCR} & \textbf{ADR} & \textbf{TRPW} & \textbf{ABM} \\ \hline \hline
      \multicolumn{1}{c|}{\multirow{2}{*}{Email}}   & E(r)   & \cellcolor{ranktwo} \SI{3.68E-2}{}     & \cellcolor{rankone} \SI{3.64E-2}{}     & \cellcolor{rankthree} \SI{5.13E-2}{}      & \SI{5.77E-2}{}    & \SI{3.36E-1}{}     & \SI{2.12E-1}{}    \\ \cline{2-8} 
      \multicolumn{1}{c|}{}                         & var(r) & \cellcolor{rankone} \SI{6.62E-4}{}     & \cellcolor{ranktwo} \SI{7.64E-4}{}     &  \SI{4.65E-3}{}      & \cellcolor{rankthree} \SI{2.64E-3}{}    & \SI{5.64E-2}{}     & \SI{2.13E-2}{}  \\ \hline
      \multicolumn{1}{c|}{\multirow{2}{*}{Flickr}}  & E(r)   &  \cellcolor{rankthree} \SI{1.36E-1}{}     & \SI{1.51E-1}{}     & \cellcolor{ranktwo} \SI{1.2E-1}{}      & \SI{1.74E-1}{}    & \cellcolor{rankone} \SI{1.03E-1}{}     & \SI{3.40E-1}{}    \\ \cline{2-8} 
      \multicolumn{1}{c|}{}                         & var(r) & \cellcolor{ranktwo} \SI{1.66E-2}{}     & \cellcolor{rankthree} \SI{2.85E-2}{}     &  \SI{3.32E-2}{}      & \SI{3.18E-2}{}    & \cellcolor{rankone} \SI{9.32E-3}{}     & \SI{7.61E-2}{}  \\ \hline
      \multicolumn{1}{c|}{\multirow{2}{*}{Biomine}} & E(r)   & \cellcolor{rankone} \SI{1.92E-2}{}     & \SI{3.77E-2}{}     & \SI{3.31E-2}{}      & \cellcolor{ranktwo} \SI{2.73E-2}{}    & \cellcolor{rankthree} \SI{3.09E-2}{}     & \SI{8.74E-2}{}   \\ \cline{2-8} 
      \multicolumn{1}{c|}{}                         & var(r) & \cellcolor{ranktwo} \SI{1.97E-4}{}     & \cellcolor{rankthree} \SI{5.80E-4}{}     & \SI{6.69E-4}{}      & \cellcolor{rankone} \SI{3.76E-5}{}    & \SI{1.14E-3}{}     & \SI{2.12E-3}{}  \\ \hline
      \multicolumn{1}{c|}{\multirow{2}{*}{DBLP}}    & E(r)   & \cellcolor{ranktwo} \SI{4.13E-2}{}     & \cellcolor{rankthree} \SI{4.23E-2}{}     & \SI{5.15E-2}{}      & \SI{4.71E-2}{}    & \SI{6.36E-2}{}     & \cellcolor{rankone} \SI{2.75E-2}{}   \\ \cline{2-8} 
      \multicolumn{1}{c|}{}                         & var(r) & \cellcolor{rankthree} \SI{5.04E-3}{}     & \cellcolor{ranktwo} \SI{4.67E-3}{}     & \SI{6.87E-3}{}      & \SI{6.47E-3}{}    & \SI{6.41E-3}{}     & \cellcolor{rankone} \SI{6.52E-4}{}  \\ 
    \end{tabular}
  }
\end{table}

\subsection{Comparison with the Exact CRPW}
\label{sec:Comparison with the Exact CRPW}

\begin{figure}[t]
  \setlength{\abovecaptionskip}{0.1cm}
    \centering
    \includegraphics[width=0.95\linewidth]{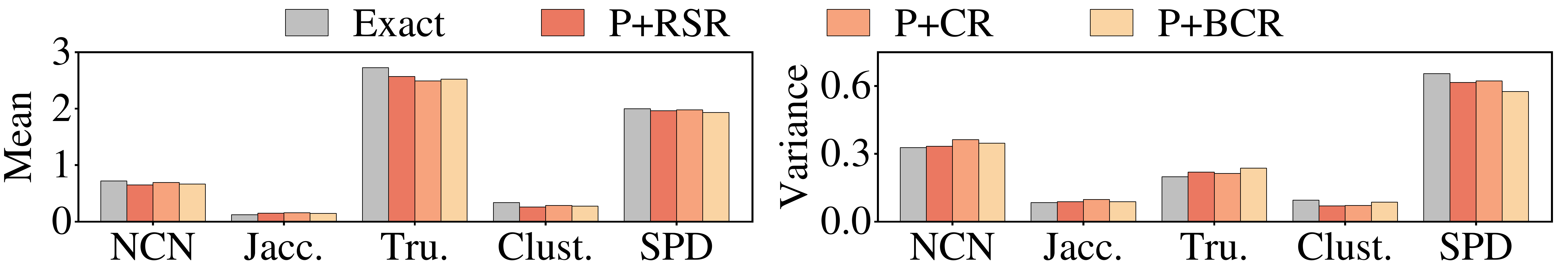}
    \vspace{-0.1cm}
    \caption{Mean and variance of structural features on Rath}
    \label{fig:Mean and variance of structural features on Rath}
    \vspace{-0.41cm}
\end{figure}

We compare our three methods with Exact only on Rath, since enumerating all \(2^{|E|}\) possible worlds is infeasible on large graphs. For each algorithm, we compute the distributions of NCN, Jacc., trussness, Clust., and SPD on the possible world it produces, and obtain the mean and variance of each distribution. As shown in Figure \ref{fig:Mean and variance of structural features on Rath}, our methods achieve means and variances close to those of Exact across the five features. This shows that our methods can also produce good representative possible worlds without exhaustive enumeration.

\subsection{Efficiency Evaluation}
\label{sec:exp_efficiency}

Figure~\ref{fig:efficiency_bar} reports the running time of all methods on the four datasets. Among the baselines, ABM is the fastest on all datasets, but at the price of its weak effectiveness (\S\ref{sec:exp_effectiveness}). ADR is also faster than our methods on most datasets, mainly because its iteration count is fixed at $10\sum_{e \in \mathcal{E}} p(e)$, far below ours on Flickr and Biomine ($1.4$M and $5.9$M versus $16.0$M and $34.0$M); on DBLP, where the two iteration counts become close ($12.4$M versus $17.2$M), this advantage disappears. TRPW is the slowest throughout: preserving both expected node degrees and expected triangle degrees incurs the heaviest per-iteration computation.

We then focus on our three methods. P+CR replaces the exact computation in each iteration of P+RSR with counting, which makes it $38.6\%$ faster than P+RSR on average. P+BCR builds on P+CR and additionally maintains the Beta distribution to terminate the refinement adaptively once the desired quality is reached. This additional maintenance raises the running time by $7.2\%$ on average over P+CR. In return, P+BCR no longer requires a manually specified iteration count, but instead lets the user specify the desired quality directly and stops on its own once that quality is attained.

\begin{figure}
  \setlength{\abovecaptionskip}{0.1cm}
    \centering
    \includegraphics[width=0.95\linewidth]{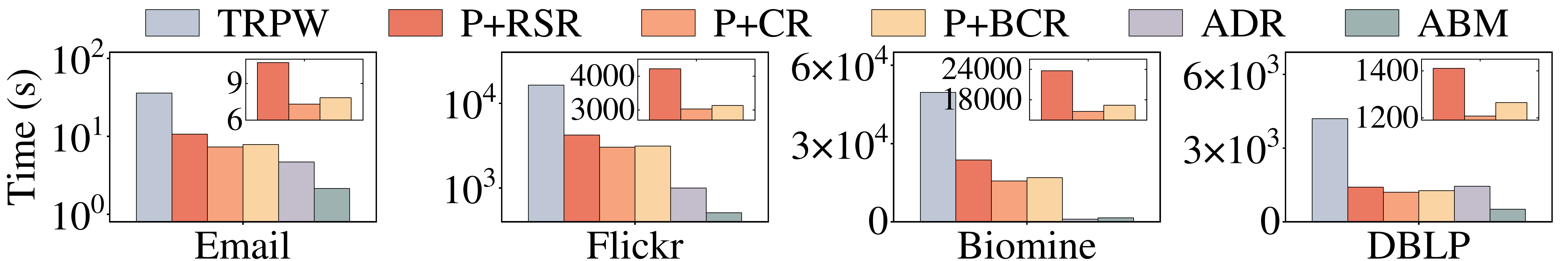}
    \caption{Running time of all methods}
    \label{fig:efficiency_bar}
    \vspace{-0.05cm}
\end{figure}

\subsection{Effectiveness of the Beta-based Termination}
\label{sec:exp_beta}

The Beta-based termination ends the counting-based refinement once the confidence that SEP $< \tau_q$ reaches $\tau_c$. Its effectiveness rests on the two following requirements. First, the confidence must be a faithful signal of refinement progress, so that the point at which it reaches $\tau_c$ indeed corresponds to a well-refined possible world; we examine this in the \textit{Reliability of the Confidence Signal}. Second, once the confidence is trustworthy, the two thresholds $\tau_q$ and $\tau_c$ must translate into predictable control over the refinement quality, so that the user can steer the output through them; we examine this in the \textit{Quality Control via Confidence Thresholds}.

\noindent\textbf{Reliability of the Confidence Signal.}
Figure~\ref{fig:confidence_iteration} reports how the confidence that SEP $< \tau_q$ evolves with the iteration count $k$ on the four datasets. The figure shows that the confidence cannot trigger the termination prematurely. For all datasets, the confidence stays negligible (as low as $10^{-200}$) for a long time, during which the refinement is still incomplete and many edges remain state-changeable. It then rises steeply, marking the point where enough state-changeable edges have been removed for the SEP to fall below $\tau_q$. The confidence is therefore a faithful signal of refinement progress, the first requirement for the termination to be effective.

\begin{figure}
  \setlength{\abovecaptionskip}{0.1cm}
    \centering
    \includegraphics[width=0.95\linewidth]{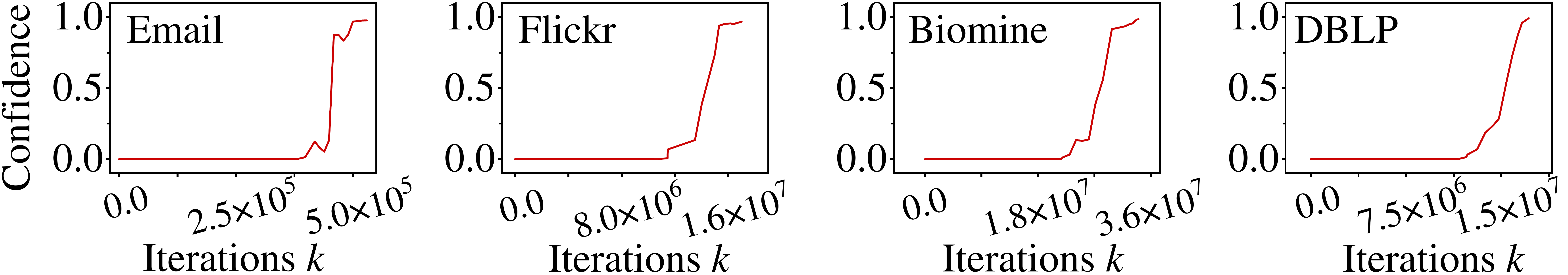}
    \caption{Confidence that SEP $< \tau_q$ as the refinement proceeds}
    \label{fig:confidence_iteration}
    \vspace{-0.05cm}
\end{figure}

\noindent\textbf{Quality Control via Thresholds.}
Figure~\ref{fig:tau_heatmap} reports the MAD of the resulting possible world under varying $\tau_q$ and $\tau_c$ on the four datasets. Across all datasets, the MAD increases with $\tau_q$ and decreases with $\tau_c$. The two trends follow from the roles of the two thresholds. $\tau_q$ is the threshold on the SEP at which the refinement may stop, so a larger $\tau_q$ permits more state-changeable edges to remain and stops the refinement earlier, leaving the possible world less refined and thus raising the MAD. $\tau_c$ is the confidence required before stopping, so a larger $\tau_c$ demands stronger evidence that SEP $< \tau_q$, which defers the stop to a more refined possible world and thus lowers the MAD. Both trends are monotonic, so a user can steer the refinement quality predictably through $\tau_q$ and $\tau_c$, which is the second requirement for the termination to be effective.

\begin{figure}
  \setlength{\abovecaptionskip}{0.1cm}
    \centering
    \includegraphics[width=0.98\linewidth]{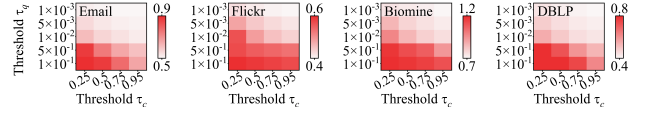}
    \caption{MAD under varying $\tau_q$ and $\tau_c$.}
    \label{fig:tau_heatmap}
    \vspace{-0.25cm}
\end{figure}

\subsection{Practical vs. Theory-Aligned Sampling}
\label{sec:Practical vs. Theory-Aligned Sampling}

\begin{figure}
  \setlength{\abovecaptionskip}{0.1cm}
    \centering
    \includegraphics[width=0.98\linewidth]{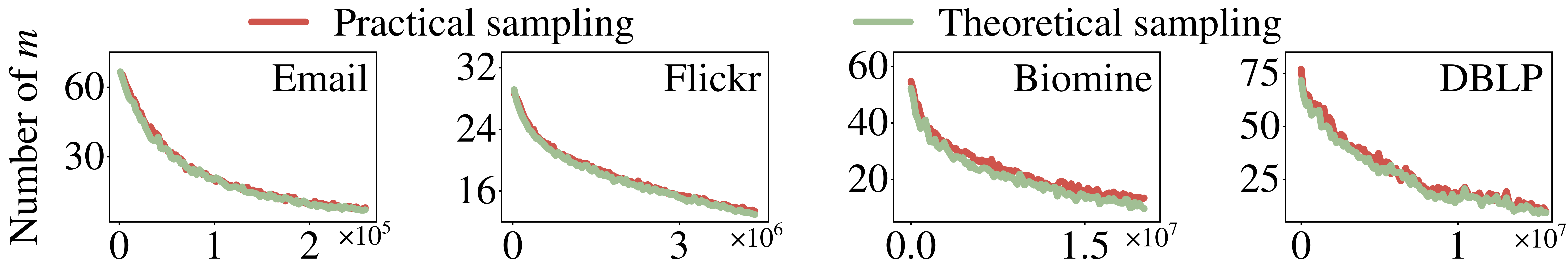}
    \caption{$m$ under practical and theory-aligned sampling.}
    \label{fig:sampling_effectiveness}
    \vspace{-0.4cm}
\end{figure}

\begin{table}[t]
  \small
  \caption{Runtime of BCR with two sampling implementations on four datasets (in seconds)}
  \label{tab:sampling_runtime}
  \setlength{\tabcolsep}{2.2pt}
  \resizebox{\columnwidth}{!}{%
    \begin{tabular}{c||c|c|c|c}
      \textbf{Datasets} $\rightarrow$
        & \textbf{Email} & \textbf{Flickr} & \textbf{Biomine} & \textbf{DBLP} \\
      \hline\hline
      Practical
        & 12.9\,(\textcolor{red}{$\downarrow$40\%})
        & 3133.8\,(\textcolor{red}{$\downarrow$43\%})
        & 16959.6\,(\textcolor{red}{$\downarrow$47\%})
        & 1264.8\,(\textcolor{red}{$\downarrow$37\%}) \\
      \hline
      Theoretical & 21.5 & 5495.6 & 31980.4 & 1998.2
    \end{tabular}%
  }
\end{table}

As noted in Sec.~\ref{sec:beta_algorithm}, BCR's practical sampling implementation (Algorithm~\ref{alg:beta_integer}) differs from theoretical one, we compare them on SEP estimation and runtime. For SEP estimation, we apply both methods within a single BCR run. Practical sampling obtain one edge per iteration and counts state-changeable samples $m$ over $n$ iterations, whereas theoretical one draws $n$ edges from the fixed possible world every $n$ iterations. Figure~\ref{fig:sampling_effectiveness} shows that practical sampling generally yields $m$ no lower than theoretical one. For fixed $n$, a larger $m$ raises SEP estimate and lowers confidence that $\mathrm{SEP}<\tau_q$, tightening termination condition; thus practical sampling is at least as conservative. For runtime (Table~\ref{tab:sampling_runtime}), we run BCR separately with each sampling. Practical sampling is faster on all datasets by reusing BCR's edge selections and state checks.

\subsection{Parameter Sensitivity}
\label{sec:exp_sensitivity}

\noindent\textbf{Effect of Iteration Count.}
We vary the iteration $k$ from $20\%$ to $100\%$ of $|\mathcal{E}|$ for each dataset and report the MAD (Fig.~\ref{fig:rsr-ir-mad}) and running time (Fig.~\ref{fig:rsr-ir-time}) of P+RSR and P+CR. 
As $k$ grows, the running time of both methods increases because each additional iteration performs one more refinement step, while their MAD decreases because the extra iterations give the refinement more chances to change edge states that reduce the total discrepancy. The MAD decrease slows as $k$ grows. This is because each iteration changes the state of a randomly selected edge, but the more iterations have been run, the fewer edges remain whose state change reduces the total discrepancy, so the reduction per iteration shrinks. Across the four datasets, P+CR raises the MAD by only $2.8\%$ on average while reducing the running time by $32.3\%$ on average, confirming that the integer rounding trades a small quality loss for a substantial efficiency gain.

\begin{figure}
  \setlength{\abovecaptionskip}{0.1cm}
    \centering
    \includegraphics[width=0.95\linewidth]{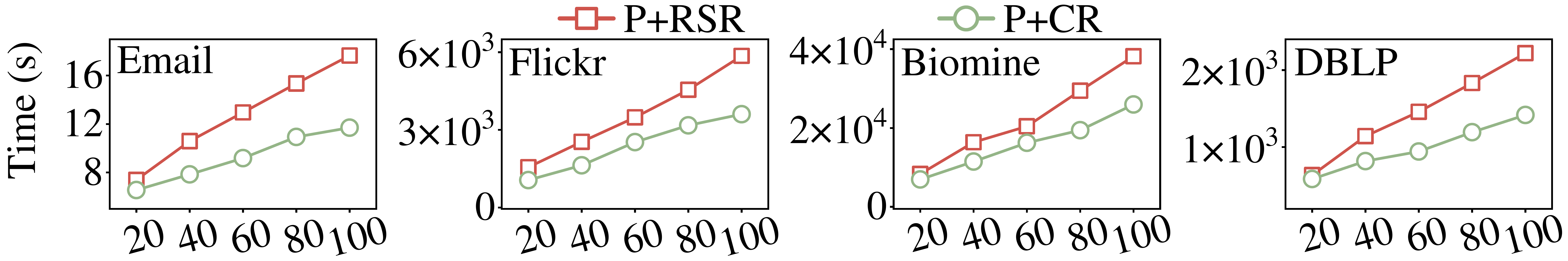}
    \caption{Time of P+RSR and P+CR under varying iteration count $k$ (from $20\%$ to $100\%$ of $|\mathcal{E}|$).}
    \label{fig:rsr-ir-time}
    \vspace{-0.05cm}
\end{figure}

\begin{figure}
  \setlength{\abovecaptionskip}{0.1cm}
    \centering
    \includegraphics[width=0.95\linewidth]{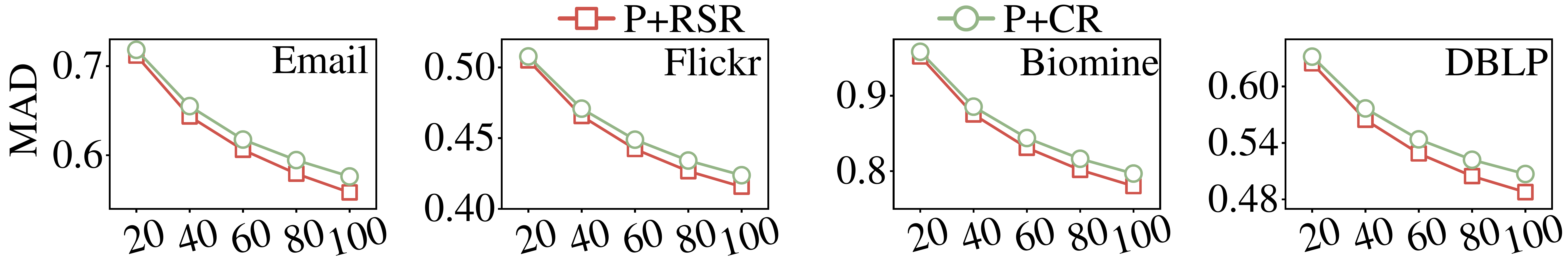}
    \caption{MAD of P+RSR and P+CR under varying iteration count $k$ (from $20\%$ to $100\%$ of $|\mathcal{E}|$).}
    \label{fig:rsr-ir-mad}
    \vspace{-0.05cm}
\end{figure}

\noindent\textbf{Effect of $\tau_q$ and $\tau_c$.}
Figure~\ref{fig:bir-time} reports the running time of P+BCR under varying $\tau_q$ and $\tau_c$ on four datasets; their effect on the MAD is examined in \S\ref{sec:exp_beta}. Across all datasets, the running time decreases as $\tau_q$ grows and increases as $\tau_c$ grows, and both trends hold regardless of the value of the other threshold. A larger $\tau_q$ relaxes the target on the SEP, so the termination criterion is met after fewer iterations, shortening the running time. A larger $\tau_c$ requires the criterion to hold more strictly before terminating, deferring the iteration at which it is first met and thus lengthening the running time.

\begin{figure}
  \setlength{\abovecaptionskip}{0.1cm}
    \centering
    \includegraphics[width=0.98\linewidth]{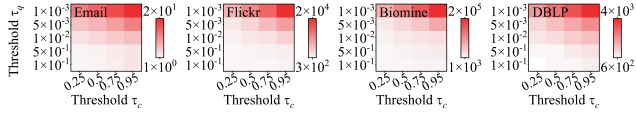}
    \caption{Time of P+BCR under varying $\tau_q$ and $\tau_c$.}
    \label{fig:bir-time}
    \vspace{-0.05cm}
\end{figure}

\section{Related Work}
\label{sec:related_work}

Related work falls into two lines: structural feature computation on uncertain graphs and representative possible-world construction.

\vspace{0.1cm}
\noindent\textbf{Structural feature computation on uncertain graphs.}
Structural feature computation on uncertain graphs is commonly formulated under the possible world model. Since exact computation by enumerating all possible worlds is prohibitive, sampling is widely used to estimate feature expectations. Potamias et al. \cite{potamias2010knn} estimate distances, while Pfeiffer and Neville \cite{pfeiffer2011methods} estimate betweenness centrality and clustering coefficients. Jin et al. \cite{jin2011distance} use recursive sampling to estimate distance-constrained reachability probabilities. Ma et al. \cite{ma2019linc} incrementally update motif counts across sampled worlds. Naito and Fushimi \cite{naito2022improving} ensemble sampled graphs to estimate motif-role frequencies. Liu et al. \cite{ACMTran:zhenfang2023} share BFS operations for closeness centrality estimation. However, sufficient accuracy can still require many samples, incurring substantial overhead.

\vspace{0.1cm}
\noindent\textbf{Representative possible worlds.}
To overcome this limitation, the RPW is introduced. Parchas et al. \cite{parchas2014pursuit} first presented the degree-based RPW, which preserves the expected degree of every node. They further extended this work to preserve the expected n-clique cardinality of each node, thereby capturing local neighborhood connectivity beyond node degrees \cite{parchas2015uncertain}. Song et al. \cite{song2016triangle} further preserve the expected triangle degree of every node, as triangles are an important structure in many graph tasks such as clustering coefficient \cite{watts1998collective}.
However, all these RPWs preserve only per-node features, while ignoring the number of common neighbors, a node pair feature. Yet the number of common neighbors is a key structure feature \cite{kogge2016jaccard, watts1998collective}. This motivates us to study the CRPW problem in this paper.

\section{Conclusion}
\label{sec:conclusion}
We study the CRPW problem, which seeks the possible world that best preserves expected common-neighbor counts, and establish its NP-hardness. We develop a two-stage algorithm combining quick initialization with iterative refinement. We accelerate the refinement by replacing its tardy exact evaluation with counting-based trend estimation, and further design a Beta-based adaptive termination that allows the refinement to stop once the user-desired quality is reached. Experiments on five real-world uncertain graphs demonstrate the effectiveness of our methods.

\bibliographystyle{ACM-Reference-Format}
\bibliography{sample}

\end{document}